\documentclass{siamltex}
\usepackage{amsmath,amsfonts,amssymb}
\usepackage{latexsym}
\usepackage{graphicx,overpic}
\usepackage{subfigure,caption}
\usepackage{pgfplots}
\pgfplotsset{compat=newest}
\usepackage{color}
\usepackage{booktabs}
\usepackage{tikz}
\usepackage{todonotes}
\usepackage[shortlabels]{enumitem}

\newcommand{\R}{\mathbb R}

\newcommand{\bB}{\mathbf B}

\newcommand{\bE}{\mathbf E}
\newcommand{\bF}{\mathbf F}
\newcommand{\bG}{\mathbf G}

\newcommand{\bI}{\mathbf I}

\newcommand{\bP}{\mathbf P}

\newcommand{\ba}{\mathbf a}

\newcommand{\bg}{\mathbf g}
\newcommand{\bn}{\mathbf n}
\newcommand{\be}{\mathbf e}
\newcommand{\bff}{\mathbf f}

\newcommand{\bT}{\mathbf T}
\newcommand{\bu}{\mathbf u}
\newcommand{\bv}{\mathbf v}
\newcommand{\bw}{\mathbf w}
\newcommand{\bx}{\mathbf x}
\newcommand{\by}{\mathbf y}
\newcommand{\bz}{\mathbf z}
\newcommand{\bbf}{\mathbf f}

\newcommand{\cG}{\mathcal G}

\newcommand{\cO}{\mathcal O}

\newcommand{\cR}{\mathcal R}
\newcommand{\cS}{\mathcal S}
\newcommand{\cL}{\mathcal L}

\newcommand{\bGamma}{\boldsymbol{\Gamma}} %Eventuell noch ändern
\newcommand{\dist}{\mathop{\rm dist}}
\newcommand{\trace}{\mathop{\rm tr}}
\newcommand{\Div}{\mathop{\rm div}}

\newcommand{\divG}{{\mathop{\,\rm div}}_{\Gamma}}

\newcommand{\G}{\Gamma}
\newcommand{\nablaG}{\nabla_{\Gamma}}
\newcommand{\nablaGT}{\nabla_{\Gamma}^T}
\newcommand{\hatnabla}{\widehat{\nabla}}
\newcommand{\hatnablaT}{\widehat{\nabla}^T}
\newcommand{\hatnablaG}{\widehat{\nabla}_{\Gamma}}
\newcommand{\hatdivG}{\widehat{\Div}_{\Gamma}}
\newcommand{\tildedivG}{\widetilde{\Div}_{\Gamma}}

\renewcommand{\div}{\textrm{div}\ \!}

\newcommand{\tr}{{\rm tr}}
\DeclareGraphicsExtensions{.pdf,.eps,.ps,.eps.gz,.ps.gz,.eps.Y}

\newcommand{\bxi}{\mbox{\boldmath$\xi$\unboldmath}}

\newtheorem{remark}{Remark}[section]
\begin{document}
\title{On derivations of evolving surface Navier-Stokes equations}
\author{Philip Brandner\thanks{Institut f\"ur Geometrie und Praktische  Mathematik, RWTH-Aachen
University, D-52056 Aachen, Germany (brandner@igpm.rwth-aachen.de).}
\and Arnold Reusken\thanks{Institut f\"ur Geometrie und Praktische  Mathematik, RWTH-Aachen
University, D-52056 Aachen, Germany (reusken@igpm.rwth-aachen.de).}
\and Paul Schwering\thanks{Institut f\"ur Geometrie und Praktische  Mathematik, RWTH-Aachen
University, D-52056 Aachen, Germany (schwering@igpm.rwth-aachen.de).}
}
\maketitle
%\tableofcontents
%
\begin{abstract}
In recent literature several derivations of incompressible Navier-Stokes type equations that model the dynamics of an evolving fluidic surface have been presented. These derivations  differ in the physical principles used in the modeling approach and in the coordinate systems in which the resulting equations are represented. This paper has overview character in the sense that we put five different derivations of surface Navier-Stokes equations into one framework. This then allows a systematic comparison of the resulting surface Navier-Stokes equations and shows that some, but not all, of the resulting models are the same. Furthermore, based on a natural splitting approach in tangential and normal components of the velocity we  show that all five derivations that we consider yield the same tangential surface Navier-Stokes equations.
\end{abstract}

\section{Introduction}
Navier--Stokes type equations posed on manifolds is  a classical topic in analysis, e.g.~\cite{ebin1970groups,Temam88,taylor1992analysis,arnold1999topological,mitrea2001navier}.
 In recent years there has been a strongly growing interest in surface Navier-Stokes equations, in particular concerning physical principles related to these equations and to tailor-made numerical discretization methods \cite{Jankuhn_Reusken_Olshanksii,Gigaetal,Hu_Zhang,Miura_17,NitschkeReutherVoigt2019,nitschke2012finite,reuther2018solving,Fries2018,NitschkeReutherVoigt2019,gross2019meshfree,BonitoDemlowLicht2020,Olsh_Reusken_2021,lederer2019divergence}. One reason for this recent growing interest lies in the fact that these equations are used in the modeling of biological interfaces, cf. the overview paper \cite{Arroyo2019} and the references therein.

In this paper we focus on \emph{derivations} of surface Navier-Stokes equations for \emph{evolving} surfaces. In the past few years several derivations  have been presented in the literature \cite{Jankuhn_Reusken_Olshanksii,Gigaetal,Hu_Zhang,Miura_17,NitschkeReutherVoigt2019}, which differ in the physical principles used in the modeling approach and in the coordinate systems in which the resulting equations are represented.  In \cite{Jankuhn_Reusken_Olshanksii,Hu_Zhang} mass and momentum conservation laws for material \emph{surfaces} are used as basic physical principles, whereas in \cite{Miura_17,NitschkeReutherVoigt2019} similar conservation laws of mass and momentum for a material \emph{volume} are used and combined with a thin film technique. In \cite{Gigaetal} the derivation is based on energy minimization principles.
Besides these differences in physical principles, there is also a difference in the representation of the resulting flow equations. In some papers, e.g. \cite{Hu_Zhang,arroyo2009,NitschkeReutherVoigt2019, NitschkeReutherVoigt2020}, local coordinate systems (curvilinear coordinates) are used, whereas in other literature \cite{Jankuhn_Reusken_Olshanksii,Gigaetal,Miura_17} the standard Euclidean basis of $\mathbb{R}^3$, in which the evolving surface is embedded, is used. Such different coordinate systems lead to  different representations of surface differential operators such as
a covariant derivative or a surface divergence, and one has to be careful
when comparing equations formulated in such different coordinate systems.
Both the local curvilinear and the global Cartesian coordinate system have attractive properties. The local coordinate system can be very useful for modeling of more complex fluid properties, e.g.  in certain classes of
fluid membranes \cite{Hu_Zhang,Arroyo2019} or in flows of  liquid crystals \cite{NitschkeReutherVoigt2019, NitschkeReutherVoigt2020}. The representation in global Cartesian coordinates is very convenient for the
development of numerical simulation methods for these flow equations.

This paper has overview character in the sense that we put the different derivations of surface Navier-Stokes equations presented in~\cite{Jankuhn_Reusken_Olshanksii,Gigaetal,Hu_Zhang,Miura_17,NitschkeReutherVoigt2019} into one framework. Besides the unified survey of derivations we also present the following (new) results:
\begin{itemize}
\item Precise relations of certain relevant differential operators, such as covariant derivatives and surface divergence operators, in different coordinate systems are given. Most of these can be found or are (implicitly) used at different places in the literature. Here we put this into one framework and derive precise relations, e.g. as in Theorem~\ref{theorem_comparison_local_cartesian} and Lemma~\ref{lemma_rate_strain_nabla_representation}.
\item The presentation in a unified framework allows a systematic comparison of the resulting surface Navier-Stokes equations. We will conclude that some of these are identical but also some are different.
\item A splitting approach in tangential and normal components of the velocity is presented, which shows that all five derivations that we consider yield  \emph{the same tangential} surface Navier-Stokes equations.
\end{itemize}
Since the (incompressible) surface Navier-Stokes equations play a fundamental role in the modeling of interfaces or surfaces with fluidic behaviour, we consider a good understanding of several known surface Navier-Stokes systems to be of major importance.

The remainder of this paper is organized as follows. In Section~\ref{secmatsurface} we define evolving material surfaces. In Section~\ref{seccoordanddiffop}  surface differential operators in different coordinate systems are defined and compared.
 Five derivations of surface Navier-Stokes equations, known from the literature, that differ in the underlying physical principles and in the coordinate systems used, are treated in Section~\ref{secderivation}. In Section \ref{secdiscussion} we discuss and compare these equations. In particular a splitting of these  equations in the tangential and the normal components is derived and it is shown that all five derivations result in the same tangential surface Navier-Stokes system.

\section{Evolving material surfaces}\label{secmatsurface}
We outline how evolving material surfaces are defined. A more precise formal description of the notion ``material'' is given in e.g. \cite{MurdochCohen79}. Let $\Gamma=\Gamma(0)$ be a smooth (at least $C^2$) connected
surface embedded in $\mathbb{R}^3$. A  material point $\bz \in \Gamma(0)$
moves in time along a trajectory with coordinates $\bx(\bz,t) \in \R^3$ and a smooth velocity field $\bv(\bx(\bz,t),t) \in \R^3$. For all $\bz \in \Gamma(0)$, the solutions of the initial value problems
\begin{align}
\begin{cases}\label{def_flow_map}
\bx(\bz,0)=\bz,\\
\frac{d}{dt} \bx(\bz,t) = \bv( \bx(\bz,t),t)
\end{cases}
\end{align}
define the evolving surface
\begin{equation} \label{defevolving}
\Gamma(t)=\{\by\in\R^3~|~\by=\bx(\bz,t),~\bz \in \Gamma(0) \}.
\end{equation}
The flow map $\Phi_t: \Gamma(0) \to \Gamma(t)$, $0\leq t \leq T$ is defined by $\Phi_t(\bz)=\bx(\bz,t)$. Let  $\Phi_U: \R^2 \supset U \to \Gamma(0)$ be a local parametrization. We assume that the mapping $\Phi_U: U \to \Phi_U(U)$ is a diffeomorphism. The coordinates in $U$ are
denoted by $\bxi= (\xi_1,\xi_2)$. Composition of $\Phi_U$ and $\Phi_t$ yields the mapping
\begin{equation} \label{mapR} R(\bxi,t):=\Phi_t(\Phi_U(\bxi)),
\end{equation}
 which gives the position of the material point $\by = \Phi_U(\bxi) \in \Gamma(t) \subset \R^3$ at time $t$. Below in Section~\ref{seccoordinate} we use $\bxi \to R(\bxi,t)$ as a (local) parametrization of $\Gamma(t)$. Note that if the flow field $\bv$ is not identically zero, this parametrization is \emph{non}-constant as a function of $t$, even if $\Gamma(t)=\Gamma(0)$ for all $t$.\\[1ex]
The outward pointing normal vector on $\Gamma(t)$ is denoted by $\bn=\bn(\by,t)$, and $\bP=\bP(\by,t)=\bI- \bn\bn^T$ is the projection on the tangential plane at $\by\in\Gamma(t)$. Below we often delete the argument $(\by,t)$ in the notation. For a vector field $\bu$ on $\Gamma(t)$ we shall use throughout this paper the notation $\bu_T=\bP\bu$  for the tangential component and $u_N=\bu\cdot\bn$ for the coordinate in normal direction, so that
\begin{equation}
\label{def_splitting}
\bu=\bu_T+u_N\bn\quad\text{on}~\Gamma(t).
\end{equation}
If in the particle velocity $\bv(\cdot,t)=\bv_T(\cdot,t)+v_N(\cdot,t)\bn(\cdot,t)$ we have $v_N(\cdot,t)=0$ on $\Gamma(t)$, there is no normal
velocity of the surface which means that the geometry of $\Gamma(t)$ is stationary and there is only a tangential particle flow field.\\[1ex]
We assume that on $\Gamma(t)$ there is a continuous strictly positive particle density distribution denoted by $\rho(\by,t)$, $\by \in \Gamma(t)$.\\[1ex]
In Section~\ref{secderivation}, based on certain physical principles we derive Navier-Stokes type equations that determine the particle velocity field $\bv$ and the density distribution $\rho$. As discussed in the introduction we will compare derivations in different coordinate systems. Therefore, in the next section we collect results concerning representations of surface differential operators in different coordinate systems, which will be used in  Section~\ref{secderivation}.

\section{Coordinate systems and surface differential  operators}\label{seccoordanddiffop}
In this section we introduce   surface differential operators in two different coordinate systems.

\subsection{Coordinate systems} \label{seccoordinate}
We treat representations of
vector fields $\bu:\Gamma\to\R^3$ and of operator valued mappings $\bT: \Gamma \to L(\R^3,\R^3)$, where $L(\R^3,\R^3)$ denotes the space of linear mappings $\R^3 \to \R^3$, in two different coordinate systems. The first one is the \emph{Cartesian coordinate} system corresponding to the standard Euclidean basis in $\R^3$, denoted by $\{\hat{\be}_1,\hat{\be}_2,\hat{\be}_3\}$. The second one is a \emph{curvilinear coordinate system}, that we introduce below.
Most of the results presented in this section are standard material that can be found in many textbooks, e.g. \cite{Ciarlet2013,Pruess_Simonett_book}.  We use  tensor notation and the Einstein summation convention in the following way: using Latin indices (i,j,k,$\dots$) we sum over $1,2,3$, using Greek indices ($\alpha,\beta,\gamma,\dots$) we sum up over $1,2$. Partial derivatives  w.r.t. the Cartesian coordinates $\xi_\alpha$ in the standard basis of $\R^2$ are denoted
by $\partial_\alpha=\frac{\partial}{\partial \xi_\alpha}$.
%A point $\by=(y_1,y_2,y_3) \in \R^3$ is represented in the standard basis denoted by $\{\hat{\be}_1=\hat{\be}^1,\cdots,\hat{\be}_3=\hat{\be}^3\}$.

In the remainder of this section we take a fixed $t$. The local parametrization of $\Gamma=\Gamma(t)$ is given by $R(\bxi)=R(\bxi,t)$,  $\bxi
\in U$. Hence, $R(U) \subset \Gamma(t)$.
We assume that this parametrization is an immersion, hence the  matrix $\begin{pmatrix} \partial_1 R(\bxi) &  \partial_2 R(\bxi) \end{pmatrix} \in \R^{3 \times 2}$ has rank two for each $\bxi \in U$. Each point $\by \in R(U)$ can be unambiguously written by $\by=R(\bxi)$
with $\bxi \in U$. The two coordinates $\xi_\alpha$ of $\bxi$ are called \textit{curvilinear} or \textit{local coordinates}
of $\by=R(\bxi)$. We introduce the covariant basis of the tangent space at $\by=R(\bxi) \in \Gamma$ given  by $\bg_\alpha=\bg_\alpha(\bxi):=\partial_\alpha R(\bxi) \in \R^3$. The components of the metric tensor (or first fundamental form) are defined by
\begin{equation} \label{firstform}
g_{\alpha \beta}(\bxi):=\bg_\alpha(\bxi)\cdot \bg_\beta(\bxi).
\end{equation}
The metric tensor is symmetric positive definite. The contravariant basis
of the tangent plane $\bg^\beta$ is defined by $\bg_\alpha \cdot \bg^\beta=\delta_\alpha^\beta$. Here, $\delta_\alpha^\beta$
denotes the Kronecker symbol. The contravariant components of the metric tensor are defined by
$
g^{\alpha \beta}(\bxi):=\bg^\alpha(\bxi) \cdot \bg^\beta(\bxi)
$.
The following relations hold:
\begin{align*}
\bg_{\alpha} = g_{\alpha\beta} \bg^\beta, \qquad \bg^{\alpha}= g^{\alpha \beta} \bg_\beta, \qquad g^{\alpha\gamma}g_{\gamma\beta}=\delta_\beta^\alpha.
\end{align*}
In order to have a basis of $\R^3$ we add to the covariant and contravariant basis a third vector, namely the normal vector (at $\by = R(\bxi)$):
\begin{align*}
\bg_3=\bg^3:=\bn=\frac{\bg_1\times \bg_2}{\|\bg_1\times \bg_2\|}=\frac{\bg^1\times \bg^2}{\|\bg^1\times \bg^2\|}.
\end{align*}
Note that, given the first fundamental form,  this determines the choice of the orientation of the normal vector $\bn$.
The vectors $\bg_i$ and $\bg^i$ for $i=1,2,3$, each form a basis of $\R^3$.
%For $\by \in \Gamma$, there exists a local parametrization $R$ such that $\by = R(\bxi)$ with a unique $\bxi \in U \subset \R^2$.
We can (locally) interpret the basis functions $\bg_i$ and $\bg^i$ as functions defined on the surface:  $\bg_i(\by):=\bg_i(R(\bxi))$, $\bxi \in U$.
For presentation purposes it is convenient to identify the (contravariant) Euclidean basis in $\R^3$ with its covariant one, i.e. $ \hat{\be}^i:=\hat{\be}_i$, $i=1,2,3$.

For a vector field $\bu:\Gamma\to\R^3$ we introduce the  representations
\begin{equation*}
\bu=u^i\bg_i= u_i\bg^i=\hat{u}_i\hat{\be}^i.
\end{equation*}
Note that $u_i=\bu\cdot\bg_i$, $u^i=\bu\cdot\bg^i$ and $\hat{u}_i=\bu\cdot\hat{\be}_i$ hold. The $u_i$ ($u^i$) are called covariant (contravariant) components or also local coordinates. The $\hat{u}_i$ are the Cartesian coordinates.
%To change the representation of a given vector field, we use the relations \todo[inline]{AR: you never refer to these relations;  we do use these?}
%\begin{equation*}
%\hat{u}_j = \hat{u}_i \hat{\be}^i \cdot \hat{\be}_j = u_i \bg^i \cdot
%\hat{\be}_j, \qquad u_j = u_i \bg^i \cdot \bg_j = \hat{u}_i \hat{\be}^i \cdot \bg_j.
%\end{equation*}
%The quantities needed for the change of basis are denoted by $[\bg^i]_j :=\bg^i \cdot\hat{\be}_j$ and $[\bg_j]^i :=  \bg_j \cdot \hat{\be}^i$. Note that $\bg^i = [\bg^i]_j \hat{\be}^j$, $\bg_j = [\bg_j]^i \hat{\be}_i$ and $\hat \be_j = \hat \be^j = [\bg^i]_j\bg_i = [\bg_i]^j
%\bg^i$ hold.

For representation of an operator valued mapping $\bT: \Gamma \to L(\R^3,\R^3)$  we use the tensor calculus format (cf. \cite[Section 8.4]{Ciarlet2013}):
\begin{equation*}
\bT = T_{ij}(\bg^i\otimes\bg^j)=T^{ij}(\bg_i\otimes\bg_j)=\hat T_{ij}(\hat{\be^i}\otimes\hat\be^j),
\end{equation*}
with the outer product given by $(\bu \otimes \bv) \bw = (\bv \cdot
\bw) \bu$ for all $\bu,\bv,\bw\in\R^3$. These representations define corresponding  matrices
%$T= (T_{ij})_{1 \leq i,j \leq 3}$, $\hat T=
%(\hat T_{ij})_{1 \leq i,j \leq 3}$,
that are representations of the same linear operator in different bases. The matrix entries satisfy  $T_{ij}=\bg_i\cdot(\bT \bg_j), ~T^{ij}=\bg^i\cdot(\bT \bg^j),~\hat{T}_{ij}=\hat{\be}_i\cdot(\bT \hat\be_j)$ and are called covariant, contravariant and Cartesian components, respectively. We define the transposed linear operator $\bT^T$ by the relation $\bT\bu \cdot \bw =\bu \cdot \bT^T \bw$ according to the Euclidean scalar product in $\R^3$. Tensors can also be represented using \emph{mixed} components (cf.  \cite[Section 8.4]{Ciarlet2013}). For a \emph{symmetric} linear operator $\bT$, i.e. $\bT = \bT^T$, we introduce the mixed (between covariant and contravariant) matrix representation by
\begin{equation*}
\bT = T^i_j(\bg_i\otimes\bg^j) = T^i_j(\bg^j\otimes\bg_i).
\end{equation*}
The relations $T^i_j=\bg^i\cdot (\bT \bg_j) = \bg_j \cdot (\bT \bg^i)$
hold.
%To change between the different matrix representations one can use the following identities:\todo[inline]{AR: you never refer to these relations;  we do use these?}
%\begin{equation*}
%\hat{T}_{ij}=[\bg^k]_i[\bg^l]_j T_{kl},\quad T_{ij}=[\bg_i]^k[\bg_j]^l\hat{T}_{kl}, \quad \hat{T}_{ij}=[\bg_k]^i[\bg^l]_j T^k_l,\quad T^i_j=[\bg^i]_k[\bg_j]^l\hat{T}_{kl}.
%\end{equation*}

For a symmetric linear operator $\bT$ the sum of its eigenvalues is denoted by $\tr(\bT)$. Since eigenvalues are invariant under basis transformations, we have  $\tr( \bT ) = T_{ii}= \hat{T}_{ii} = T_i^i$.

The projection operator ${\bP=\bI-\bn\bn^T}$ is defined in local coordinates by $\bP(c^i\bg_i):=c^\alpha\bg_\alpha$. In local coordinates  the splitting  as in \eqref{def_splitting} takes the form  $\bu_T:=\bP\bu=u^\alpha\bg_\alpha$ and $u_N=\bu\cdot\bn= u^3 = u_3$.

We recall  the second fundamental form $\bB = \bB(\by)$, $\by \in \Gamma$, also called Weingarten mapping or shape operator, which is a symmetric linear operator for which $\bB = \bP\bB \bP$ holds. From the latter property it follows that $\bB= b_{ij}(\bg^i \otimes \bg^j)=b_{\alpha \beta}(\bg^\alpha \otimes \bg^\beta)$ holds. For the covariant components  we have \cite[Theorem 8.13-1, Theorem 8.14-1]{Ciarlet2013}
\begin{align}\label{eqn_def_weingarten_covariant}
b_{\alpha \beta}&= \bg_3 \cdot \partial_\alpha \bg_\beta = -\partial_\alpha \bg_3 \cdot \bg_\beta =  b_{\beta \alpha}.
\end{align}
For the mixed components of the second fundamental form  the relation $ b_{\alpha}^\beta = g^{\beta\sigma} b_{\sigma\alpha}$ holds. %Note that they are in general not symmetric.
Let $\kappa_1,\kappa_2$ and $0$ be the eigenvalues of  $\bB$.
We introduce the (doubled) mean curvature $\kappa = \tr(\bB)=\kappa_1 + \kappa_2$ and the Gaussian curvature $K=\kappa_1 \kappa_2$. The mean curvature can be represented in mixed components by $\kappa = b^\alpha_\alpha$.

\subsection{Surface differential operators} \label{secdiffop}
In this section we recall several surface differential operators. For a given $t$, let    $\phi:\Gamma= \Gamma(t) \to \R$ be a be a scalar function, $\bu:\Gamma \to \R^3$ be a (not necessarily tangential) vector field and $\bT: \Gamma \to L(\R^3,\R^3)$ an operator valued mapping. All are assumed to be at least $C^1$-smooth. We will study partial derivatives and gradients of $\phi$, $\bu$ and $\bT$ and divergence operators for  $\bu$ and $\bT$. Representations in different bases of $\bu=\bu(\by)$ and $\bT = \bT (\by)$, $\by \in \Gamma$ are considered. First in Section~\ref{sectlocal} we recall standard definitions and results for derivatives in local coordinates representation. Note that in this case the basis used in $\R^3$ depends on the (base) point $\by$.
In Section~\ref{sectCartesian} we list (standard) definitions for analogous gradient and divergence operators in case of representation in Cartesian coordinates in $\R^3$. In Section~\ref{sectrelations} we then derive relations between the corresponding operators in the different representations. In the last part of this section we introduce the material derivative in a direction along the moving surface, which can also be formulated both in curvilinear  and Cartesian coordinates.

\subsubsection{Surface differential operators in curvilinear coordinates} \label{sectlocal}
We recall some basic differential geometry concepts, e.g. \cite[Chapter 8]{Ciarlet2013}. Note that we have the following representations in curvilinear coordinates: $\bu=u^i\bg_i=u_i\bg^i$ and $\bT = T^{ij}(\bg_i\otimes\bg_j) = T_{ij}(\bg^i \otimes \bg^j)$. All component functions are differentiable because the basis vectors $\bg_i$ and $\bg^i$ are smooth. Since $R$
is an immersion, there are uniquely defined functions $\bar{\phi}:U \to \R$, $\bar{\bu}:U \to \R^3$ and $\bar{\bT}: U \to L(\R^3,\R^3)$ such that $\bar{\phi}(\bxi)=\phi(R(\bxi))$, $\bar{\bu}(\bxi)=\bu(R(\bxi))$ and $\bar{\bT}(\bxi)=\bT(R(\bxi))$.
\begin{definition}
The \emph{partial} derivatives $\partial_\alpha$ of the scalar function $\phi$, the vector field $\bu$ and the linear operator $\bT$ are defined in terms of the corresponding functions $\bar{\phi}$, $\bar{\bu}$ and $\bar{\bT}$ by
\begin{align*}
\partial_\alpha \phi(\by):=\partial_\alpha\bar{\phi}(\bxi), \quad \partial_\alpha \bu(\by):=\partial_\alpha\bar{\bu}(\bxi), \quad \partial_\alpha \bT(\by):=\partial_\alpha\bar{\bT}(\bxi) \quad \text{with } \by=R(\bxi).
\end{align*}
\end{definition}
We now derive representations of these partial derivatives  in terms of a curvilinear coordinate system, which are used at several places in the remainder of this paper.
For this we use the Christoffel symbols (cf. \cite[Theorem 8.13-1]{Ciarlet2013})
\begin{equation*}
\Gamma_{\alpha\beta}^\sigma:=\bg^\sigma\cdot\partial_\alpha\bg_\beta =
\Gamma_{\beta \alpha}^\sigma.
\end{equation*}
These symbols can also be formulated in terms of the metric tensor (cf. \cite[Theorem 8.13-1,Theorem 8.14-1]{Ciarlet2013}):
\begin{equation*}
\Gamma_{\alpha\beta}^\sigma = \dfrac{1}{2} g^{\sigma\tau}(\partial_\beta g_{\alpha\tau} + \partial_\alpha g_{\beta\tau} -\partial_\tau g_{\alpha\beta}).
\end{equation*}
Representations of partial derivatives of vector fields in terms of a curvilinear coordinate system are given in  the following theorem from \cite[Theorem 8.13-1]{Ciarlet2013}. We extend this theorem with an analogous result \eqref{extra} for partial derivatives of operator-valued functions $\bT: \Gamma \to L(\R^3,T\Gamma)$, where $T\Gamma$ denotes the tangent bundle of $\Gamma$. A proof of the  result  \eqref{extra} is given in the appendix.

\begin{theorem} \label{thm_partial_derivatives}
	% The derivatives of the vectors of the covariant and contravariant bases are given by
	% \begin{equation*}
 	% \partial_\alpha \bg_\beta = \Gamma_{\alpha\beta}^\sigma \bg_\sigma + b_{\alpha\beta} \bg_3,~ \partial_\alpha \bg^\beta = -\Gamma_{\alpha\sigma}^\beta \bg^\sigma + b_{\alpha}^\beta\bg^3 ,~ \partial_\alpha \bg_3 =
%\partial_\alpha \bg^3 = -b_{\alpha\beta} \bg^\beta = -b_{\alpha}^\sigma \bg_\sigma.
	% \end{equation*}
For a vector field $\bu$, the partial derivatives have the following representations:
\begin{equation} \label{eqn_partial_deriv_u}
\begin{aligned}
	\partial_\alpha \bu = \partial_\alpha (u_i \bg^i) &= (\partial_\alpha u_\beta - \Gamma_{\alpha\beta}^\gamma u_\gamma - b_{\alpha\beta} u_3) \bg^\beta + (\partial_\alpha u_3 + b_\alpha^\beta u_\beta ) \bg^3 \\
	&= (u_{\beta|\alpha} -b_{\alpha\beta}u_3)\bg^\beta + (u_{3|\alpha} + b^\beta_\alpha u_\beta) \bg^3 \\
	= \partial_\alpha (u^i\bg_i) &= (\partial_\alpha u^\beta +  \Gamma_{\alpha\gamma}^\beta u^\gamma - b_\alpha^\beta u^3) \bg_\beta + (\partial_\alpha u^3 + b_{\alpha\beta} u^\beta ) \bg_3 \\
	&=(u^\beta_{~|\alpha} - b^\beta_{\alpha} u^3) \bg_\beta + (u^3_{~|\alpha} + b_{\beta\alpha} u^\beta) \bg_3,
\end{aligned}
\end{equation}
where we use the abbreviations
	\begin{align*}
	u_{\beta|\alpha} := \partial_\alpha u_\beta - \Gamma_{\alpha\beta}^\gamma u_\gamma, \quad u^\beta_{~|\alpha} := \partial_\alpha u^\beta + \Gamma_{\gamma\alpha}^\beta u^\gamma, \quad u_{3|\alpha}=u^3_{~|\alpha}:=\partial_\alpha u_3.
\end{align*}

Let $\bT=T^{\alpha\beta} (\bg_\alpha \otimes \bg_\beta) = T_{\alpha\beta} (\bg^\alpha \otimes \bg^\beta)$ be a function with values in $L(\R^3,T\Gamma)$. For the partial derivatives we have the representations:
	\begin{equation} \label{extra} \begin{split}
	\partial_\gamma\bT &= T^{\alpha\beta}_{~|\gamma} (\bg_\alpha \otimes \bg_\beta) + T^{\alpha\beta} b_{\gamma\alpha} (\bg_3 \otimes \bg_\beta) + T^{\alpha\beta} b_{\gamma\beta} (\bg_\alpha \otimes \bg_3)\\
	&= T_{\alpha\beta|\gamma} (\bg^\alpha \otimes \bg^\beta) + T_{\alpha\beta} b_\gamma^\alpha (\bg^3 \otimes \bg^\beta) + T_{\alpha\beta} b_\gamma^\beta (\bg^\alpha \otimes \bg^3),
	\end{split} \end{equation}
		where we use the abbreviations
\begin{align*}
	T^{\alpha\beta}_{~|\gamma} := \partial_\gamma T^{\alpha\beta} + \Gamma_{\mu\gamma}^\alpha T^{\mu\beta} + \Gamma_{\mu\gamma}^\beta T^{\alpha\mu}, \quad T_{\alpha\beta|\gamma} := \partial_\gamma T_{\alpha\beta} -  \Gamma_{\alpha\gamma}^\mu T_{\mu\beta} - \Gamma_{\beta\gamma}^\mu T_{\alpha\mu}.
	\end{align*}
\end{theorem}
The relation $\bP\partial_\alpha \bu = u_{\beta|\alpha}\bg^\beta$ for tangential vector fields $\bu$ motivates the notation $u_{\beta|\alpha}$.
\\
We recall standard definitions of surface differential operators in curvilinear coordinates  \cite{Pruess_Simonett_book,Pruess_Simonett_KVF}.

\begin{definition} \label{Def1}
For a scalar function $\phi\in C^1(\Gamma,\R)$  the surface gradient is defined by
\begin{align*}
\nablaG \phi :=  \partial_\alpha \phi \, \bg^\alpha.
\end{align*}
For a vector field $\bu\in C^1(\Gamma,\R^3)$ we define the  $\alpha$-th partial covariant derivative $\nabla_\alpha\bu$ and the covariant derivative $\nablaG \bu$ by
\[
 \nabla_\alpha\bu := \bP\partial_\alpha\bu, \quad \nablaG\bu:=\nabla_\alpha\bu\otimes\bg^\alpha.
\]
 The surface divergence of $\bu\in C^1(\Gamma,\R^3)$ and  $\bT \in C^1(\Gamma,L(\R^3, \R^3))$ are  defined by
\begin{equation} \label{defdiv}
\divG\bu := \partial_\alpha\bu \cdot \bg^\alpha, \quad \divG\bT:=(\partial_\alpha\bT)^T\bg^\alpha.
\end{equation}
\end{definition}

Note that there is a transpose in the definition of $\divG\bT$. The definitions of the surface gradient, covariant derivative and surface divergence operators above do not depend on the choice of the parametrization, cf. \cite{Pruess_Simonett_book}.

\begin{remark} \label{remark_surface_gradient} \rm
 Another surface differential operator for a vector field $\bu\in C^1(\Gamma,\R^3)$ that plays a natural role in this setting is the \emph{surface gradient} of $\bu$, defined by  $\nabla_S \bu :=
\bg^\alpha \otimes \partial_\alpha \bu$, cf. \cite{Pruess_Simonett_book}. Note that it maps into the tangent bundle. It is related to the covariant derivative via $\nablaG \bu = \bP \nabla_S^T \bu$ (we use the notation $\nabla_S^T \bu= (\nabla_S\bu)^T$). We use this surface gradient only in the proof of Theorem \ref{theorem_comparison_local_cartesian}.
\end{remark}

\begin{remark} \rm
In \cite{Pruess_Simonett_KVF,Pruess_Simonett_book}, the covariant derivative of  \emph{tangential} vector functions $\bu$ is defined by $\nablaG\bu:=\nabla_\alpha\bu \otimes \bg^\alpha$. In Definition~\ref{Def1} we extended this to general (not necessary tangential) vector fields.
\end{remark}

Using  results from Theorem~\ref{thm_partial_derivatives} one obtains representations of the $\alpha$-th partial covariant derivative in the covariant basis $\bg_\alpha$ in terms of (derivatives of) the contravariant components $u^i$ in $\bu= u^i \bg_i$:

\begin{equation}
\label{eq_partial_covariant_derivative_component}
\nabla_\alpha\bu = \bP\partial_\alpha\bu = (\partial_\alpha u^\beta +  \Gamma_{\alpha\gamma}^\beta u^\gamma - b_\alpha^\beta u^3) \bg_\beta =
(u^\beta_{~|\alpha} - b^\beta_{\alpha} u^3) \bg_\beta.
\end{equation}

This result shows that the notation $u^\beta_{~|\alpha}$ introduced in Theorem~\ref{thm_partial_derivatives} is natural, in the sense that for tangential $\bu$ we have $\nabla_\alpha\bu=u^\beta_{~|\alpha}\bg_\beta$.

If $\bu$ is tangential, the relation
\begin{equation}
\label{eq_covariant_derivative_component}
\nablaG\bu= u_{\alpha|\beta}(\bg^\alpha \otimes \bg^\beta)
\end{equation}
holds, i.e., the covariant components of $\nablaG\bu$ are given by $u_{\alpha|\beta}$. This induces an equivalent alternative  definition of the covariant derivative, that is sometimes used in the literature. An alternative definition of surface divergence of a vector field can be based on the relation
\begin{equation}
\label{eq_divergence_vector_component}
\divG\bu=u^\alpha_{~|\alpha}.
\end{equation}

In the following lemma, we present an analogous representation result for the divergence of an operator-valued function. A proof is given in the appendix.
\begin{lemma}
\label{lemma_divergence_local_representation}
For
$\bT =  T^{\alpha \beta} (\bg_\alpha \otimes \bg_\beta)$ the following holds:
\begin{align*}
	\divG \bT = {T}^{\alpha\beta}_{~|\alpha} \bg_\beta + T^{\alpha\beta} b_{\alpha\beta}\bg_3. % =  T_{\alpha\beta|\gamma} g^{\alpha\gamma} \bg^\beta + T_{\alpha\beta} b^{\alpha\beta} \bg^3.
\end{align*}
\end{lemma}

\subsubsection{Surface differential operators in Cartesian coordinates} \label{sectCartesian}
We recall  definitions of surface differential operators in terms of representations in Cartesian coordinates as in \cite{Jankuhn_Reusken_Olshanksii}. The partial derivatives w.r.t. the standard basis $\hat{\be}_1,\hat{\be}_2,\hat{\be}_3$ in $\R^3$,  i.e., $\by= y_i \hat{\be}_i$.
are denoted  by $\hat{\partial}_k:=\frac{\partial}{\partial y_k}$. The gradient of a scalar
function $f$ w.r.t. the Cartesian coordinates is given by the  vector $\hatnabla f: = (\hat{\partial}_i f) \hat{\be}_i$. The gradient (Jacobian) of a vector-valued function $\bu$ is given by $\hatnabla \bu := \hat \partial_k \bu \otimes \hat \be_k$, or in matrix notation  $(\hatnabla  \bu)_{ij}=\hat{\partial}_j \bu_i$. Note the structural analogy between $\hatnabla \bu = \hat \partial_k \bu \otimes \hat \be_k$ and the definition of the covariant derivative $\nablaG = \bP\partial_\alpha\bu \otimes \bg^\alpha$ (cf. Definition~\ref{Def1}).

To define Cartesian  \emph{surface} differential operators based on Cartesian representations, we extend functions defined on the surface to a small open neighborhood
$
G_\delta(\G):=G_\delta:=\{\bx \in \R^3 ~|~ \dist(\bx,\G)<\delta \}
$
with some sufficiently small $\delta > 0$.
For a given scalar function $\phi$ on $\Gamma$ a smooth extension to a function defined on $G_\delta$ is denoted by $\phi^e$. Similarly for vector fields and operator-valued functions on $\Gamma$. The specific choice of the  extension is not essential; one may use a constant extension along normals. We now introduce surface differential operators based on  the ``Cartesian gradient'' $\hatnabla$, applied to the extended quantities.
\begin{definition} \label{Def2}
For a scalar function $\phi\in C^1(\Gamma,\R)$  the surface gradient is defined by
\begin{align*}
\hatnablaG \phi := \bP \hatnabla \phi^e.
\end{align*}
For a vector field $\bu\in C^1(\Gamma,\R^3)$ we define the covariant derivative $\hatnablaG \bu$  by
\[
 \hatnablaG \bu := \bP \hatnabla \bu^e \bP.
\]
 The surface divergence of $\bu\in C^1(\Gamma,\R^3)$ and  $\bT \in C^1(\Gamma,L(\R^3, \R^3))$ are  defined by
\begin{align*}
 \hatdivG \bu := \tr (\bP \hatnabla \bu^e \bP), \quad \hatdivG\bT:=\hatdivG
\big( \bT^T\hat{\be}_i \big) \hat{\be}_i.
\end{align*}
\end{definition}
These definitions of the surface differential operators in Cartesian coordinates are independent of the choice of the extension and only depend on the function values on the surface.
Note that the definition of the surface divergence of the operator valued function $\bT$ in Cartesian coordinates is based on the surface divergence of the vector field $\bT^T\hat{\be}_i$. In matrix notation this means that we take the surface divergence of $\bT$ row-wise,   which agrees with the usual definition in the literature (cf. \cite{Jankuhn_Reusken_Olshanksii, Fries2018, BonitoDemlowLicht2020}).

\subsubsection{Relations between the surface differential operators in different coordinate systems} \label{sectrelations}
In this section we derive relations between surface differential operators given in the Definitions~\ref{Def1} and \ref{Def2}.
The results are as expected and have been used (implicitly) at several places in the literature. We did not find, however, proofs of all these basic results in the literature. Therefore we include elementary proofs here.

\begin{theorem}\label{theorem_comparison_local_cartesian}
Let $\phi \in C^1(\G,\R)$, $\bu \in C^1(\G,\R^3)$ and $\bT \in C^1(\G,L(\R^3,\R^3))$.  For the surface gradients, covariant derivatives and surface divergence operators defined in Definitions~\ref{Def1} and \ref{Def2} the following relations hold on $\Gamma$:
\begin{equation} \label{mresult}
\nablaG\phi =\hatnablaG\phi, \quad \nablaG\bu=\hatnablaG\bu, \quad \divG \bu = \hatdivG \bu, \quad \divG \bT =  \hatdivG (\bT^T).
\end{equation}
\end{theorem}
\begin{proof}
A proof of the first equality can be found e.g. in \cite{Pruess_Simonett_book,DEreview}. For completeness we include an elementary proof. Using the chain rule we get, with $\by=R(\bxi) \in \Gamma $,
\begin{align*}
	\partial_\alpha \phi (\by) &= \partial_\alpha ( \phi \circ R ) (\bxi) = \partial_\alpha ( \phi^e \circ R )(\bxi) = \hat{\partial}_k\phi^e(R(\bxi)) \, (\partial_\alpha R(\bxi) \cdot \hat{\be}_k) \\
	&= \hat{\partial}_k\phi^e(\by)\,(\bg_\alpha\cdot\hat{\be}_k).
\end{align*}
Thus we get
\begin{align*}
\nablaG\phi(\by)&=\partial_\alpha\phi(\by)\bg^\alpha=\big[\hat{\partial}_k\phi^e(\by)\,(\bg_\alpha\cdot\hat{\be}_k)\big]\bg^\alpha=\hat{\partial}_k\phi^e(\by)\big[\underbrace{(\bg_\alpha\cdot\hat{\be}_k)\bg^\alpha }_{\bP\hat{\be}_k}\big]\\
&= \hat{\partial}_k\phi^e(\by)\bP\hat{\be}_k = \bP\big[\hat{\partial}_k\phi^e(\by)\hat{\be}_k\big]=\bP \hatnabla \phi^e(\by).
\end{align*}
For vector fields the transposed Jacobian is given by
% \begin{align*}
% 	\,[\hatnabla\bu^e]_{ij} = \hat\partial_j\bu^e\cdot\hat\be_i =\delta_{jk}\hat{\partial}_k\bu^e\cdot\hat{\be}_i=(\hat{\be}_j\cdot\hat{\be}_k)(\hat{\partial}_k\bu^e\cdot\hat{\be}_i) = \hat{\be}_j\cdot\left[(\hat{\be}_k\otimes\hat{\partial}_k\bu^e)\hat{\be}_i \right].
% \end{align*}
\begin{align}\label{proof_vector_function_help_1}
	\hatnablaT \bu^e = \hat{\be}_k \otimes \hat{\partial}_k \bu^e.
\end{align}
Using the chain rule we get, with $\by = R(\bxi) \in \Gamma$,
\begin{align*}
	\partial_\alpha \bu(\by) \cdot \hat{\be}_i &= \partial_\alpha (\bu \circ R) (\bxi) \cdot \hat{\be}_i = \big[ \hat\partial_k \bu^e(R(\bxi))\cdot \hat\be_i \big] \big[ \partial_\alpha R(\bxi) \cdot \hat{\be}_k \big]\\
	&= \big[ \hat\partial_k \bu^e(\by) \cdot \hat\be_i \big] \big[ \bg_\alpha \cdot \hat{\be}_k \big].
\end{align*}
 Combing this with \eqref{proof_vector_function_help_1} and using the surface gradient $\nabla_S \bu = \bg^\alpha \otimes \partial_\alpha \bu$  (cf. Remark~\ref{remark_surface_gradient}) we obtain
\begin{equation} \label{proof_comparison_help_2}
\begin{aligned}
	\nabla_S \bu \, \hat{\be}_i &= (\bg^\alpha \otimes \partial_\alpha \bu) \hat{\be}_i = \bg^\alpha (\partial_\alpha \bu \cdot \hat{\be}_i) = \bg^\alpha \big[ (\hat\partial_k \bu^e \cdot \hat\be_i) (\bg_\alpha \cdot \hat{\be}_k) \big]\\
	&= \big( \hat\partial_k \bu^e \cdot \hat\be_i \big) \bP \hat{\be}_k = \bP \big(\hat{\be}_k \otimes \hat{\partial}_k \bu^e \big) \hat\be_i = \bP \hatnablaT \bu^e \hat\be_i.
\end{aligned}
\end{equation}
Using $\nabla_\Gamma \bu= \bP \nabla_S^T\bu$ completes the proof of the relation for the covariant derivative of $\bu$.
For the surface divergence  of a vector function we get
\begin{align*}
	\divG\bu &= \partial_\alpha\bu \cdot \bg^\alpha = \delta_\gamma^\alpha (\partial_\alpha\bu \cdot \bg^\gamma) = (\bg^\alpha \cdot \bg_\gamma) (\partial_\alpha\bu \cdot \bg^\gamma) = [(\bg^\alpha \otimes \partial_\alpha\bu) \bg^\gamma] \cdot \bg_\gamma\\
	&= (\nabla_S \bu\, \bg^\gamma) \cdot \bg_\gamma \overset{\eqref{proof_comparison_help_2}}{=} (\bP \hatnabla ^T \bu^e \bg^\gamma) \cdot \bg_\gamma = (\bP \hatnabla ^T \bu^e \bP \bg^i) \cdot \bg_i = \trace( \bP \hatnabla \bu^e \bP) \\
	&= \hatdivG \bu.
\end{align*}
For the surface divergence of an operator-valued function $\bT$ we have
\begin{align*}
	\divG \bT \cdot \hat{\be}_i &= ((\partial_\alpha \bT)^T \bg^\alpha) \cdot \hat{\be}_i = \delta_\beta^\alpha \hat{\be}_i \cdot ((\partial_\alpha \bT)^T \bg^\beta) = (\bg^\alpha \cdot \bg_\beta)(\partial_\alpha \bT \hat{\be}_i) \cdot \bg^\beta\\
	& = ((\bg^\alpha \otimes \partial_\alpha \bT \hat{\be}_i) \bg^\beta ) \cdot \bg_\beta = ((\bg^\alpha \otimes \partial_\alpha \bT \hat{\be}_i) \bP \bg^i) \cdot \bP \bg_i \\
	& = \trace \big[\bP (\bg^\alpha \otimes \partial_\alpha \bT \hat{\be}_i) \bP \big] = \trace\big[ \bP \nabla_S (\bT \hat{\be}_i) \bP\big] \overset{\eqref{proof_comparison_help_2}}{=} \trace\big[ \bP \hatnabla^T (\bT^e \hat{\be}_i) \bP \big] \\
	&= \trace\big[ \bP \hatnabla (\bT^e \hat{\be}_i) \bP \big] = \hatdivG( \bT \hat{\be}_i )= \hatdivG(\bT^T)\cdot \hat{\be}_i.
\end{align*}
\end{proof}

Note that in the relation for the surface divergence of $\bT$ in \eqref{mresult} a \emph{transpose} is needed.  This would vanish if either in Definition~\ref{Def1} or in Definition~\ref{Def2} one deletes the transpose in the definition of the surface divergence of $\bT$. The results in Theorem~\ref{theorem_comparison_local_cartesian} confirm that the operators defined in Definition~\ref{Def1} indeed do not depend on the parametrization.

The shape operator, given in curvilinear coordinates in \eqref{eqn_def_weingarten_covariant}, can be represented in the Cartesian coordinate system as $\bB = - \hatnablaG \bn^e$ (proof in the appendix).

\subsubsection{The material derivative} \label{secdiffop_subsectmaterial}
We introduce a derivative in which the time dependence of the parametrization $R(\bxi,t)$, $\bxi \in U$,  is used.
Let $I=(0,T)$ be a time interval with $T >0$ sufficiently small such that for all $\bz \in \Phi_U(U) \subset \G(0)$ the ordinary differential equation \eqref{def_flow_map} has a unique solution for $t \in I$.
We define the (local) evolving surface $\Gamma_U(t)= \{\, \by \in \R^3~|~ \by = R(\bxi,t),~\bxi \in U\,\}$, $t \in I$.   The corresponding  space-time manifold is given by
\begin{align*}
\cS= \cS(U,I):=\bigcup_{t\in I} \Gamma_U(t) \times \{t\} \subset \R^4.
\end{align*}
Note that $\cS$ is parametrized by $\cR: U\times I \to \cS$, $\cR(\bxi,t)=(R(\bxi,t),t)$.
Given the velocity $\bv(\by,t)$, $\by \in \Gamma_U(t)$ from \eqref{def_flow_map} we define  $\bar{\bv}(\bxi,t):=\bv(R(\bxi,t),t)$, $(\bxi,t) \in U \times I$.
Thus we have the relation
\begin{align}
\label{relation_v_R}
	\bar{\bv}(\bxi,t) = \frac{\partial}{\partial t} R(\bxi,t) \quad \text{on}~U \times I.
\end{align}

\begin{definition}
Let $f \in C^1(\cS)$ be a scalar- or vector function and $\bar{f}\in C^1(U \times
I)$ the function defined by $\bar{f}(\bxi,t)=f(R(\bxi,t),t)$ for $(\bxi,t) \in U \times I$. The \emph{material} derivative of $f$ on $\cS$ is defined by
\begin{align*}
\overset{\bullet}{f}(\by,t) := \partial_t \bar{f}(\bxi,t), \quad \by = R(\bxi,t).
\end{align*}
\end{definition}
Clearly this is a definition in terms of the local coordinates $\bxi$ of the surface $\Gamma(0)$.

To obtain a Cartesian representation of the material derivative, we use the same approach as in the previous section and extend the functions defined on the space-time manifold $\cS$ to an open neighborhood $\cG_\delta=\cG_\delta (\cS)$, given by
$
	\cG_{\delta}= \bigcup_{t\in I} G_\delta (\Gamma_U(t)) \times \{ t\}
$.
The neighborhood $G_\delta (\Gamma_U(t))$ of $\G_U(t)$ is as defined in the previous section.
The following lemma yields a representation of the material derivative defined above in terms of derivatives w.r.t. Cartesian coordinates in $\R^3 \times \R$. The result is well-known and easy to prove, based on application of the chain rule. For completeness we include an elementary proof.

\begin{lemma}\label{lemma_material_derivative_cartesian}
Let $\phi \in C^1(\cS,\R)$ and $\bu \in C^1(\cS,\R^3)$ with smooth extensions $\phi^e \in C^1(\cG_\delta,\R)$ and $\bu^e \in C^1(\cG_\delta,\R^3)$.  For the material derivatives of $\phi$ and $\bu$ the following holds:
\begin{equation*}
\overset{\bullet}{\phi}(\by,t) = \partial_t \phi^e(\by,t) + \hatnabla \phi^e (\by,t) \cdot \bv(\by,t), \quad \overset{\bullet}{\bu}(\by,t) =
\partial_t \bu^e(\by,t) + \hatnabla \bu^e (\by,t) \bv(\by,t).
\end{equation*}
\end{lemma}
\begin{proof}
For $(\by,t) \in \cS$ we write $\by =R(\bxi,t)$.
 We use the chain rule for the function $\phi^e (R(\bxi,t),t) = (\phi^e \circ F)(t)$ with the auxiliary function $F: I \to \R^4,~t \mapsto (R(\bxi,t),t)$ and get
 % We interpret the space $\R^4$ as $\R^4 = \{\bz \in \R^{3+1}~|~ \bz = (\bx,t),~ \bx \in \R^3 \}$ with Euclidean basis $\hat{\ba}_1,\cdots,\hat{\ba}_4$. The derivation of the auxiliary function is given by $\frac{d}{dt} F(t) = (\partial_t R(\bxi,t),1) = (\bar\bv(\bxi,t), 1) = (\bv(\bxi,t),1)$.
%The chain rule yields
\begin{align}
	\overset{\bullet}{\phi} (\by,t)
	&= \partial_t \bar{\phi}(\bxi,t) = \frac{d}{dt} \phi (R(\bxi,t),t) \nonumber \\
	&= \sum_{k=1}^3 \hat{\partial}_k \phi^e(R(\bxi,t),t) \left( \frac{\partial}{\partial t} R(\bxi,t) \cdot \hat{\be}_k \right) + \partial_t \phi^e(R(\bxi,t),t) \cdot 1 \nonumber \\
	&= \partial_t \phi^e(R(\bxi,t),t) + \hatnabla \phi^e(R(\bxi,t),t) \cdot \frac{\partial}{\partial t} R(\bxi,t). \label{material_derivative_R}
\end{align}
Using $\by = R(\bxi,t)$ and relation \eqref{relation_v_R}, we obtain
\begin{equation*}
\overset{\bullet}{\phi} (\by,t) = \partial_t \phi^e(\by,t) + \hatnabla \phi^e(\by,t) \cdot \hat{\bv}(\by,t).
\end{equation*}
The same arguments can be used to derive the relation for $\bu$.
\end{proof}

The material derivative is used, for example, in the Leibniz rule or transport theorem for an arbitrary material subdomain $\gamma(t)\subset \Gamma_U(t)$:
\begin{align}
\label{eqn_Leibniz_rule}
\frac{d}{dt} \int_{\gamma(t)} f~ds=\int_{\gamma(t)} \overset{\bullet}{f}+ f \divG\bv ~ds,
\end{align}
for $f \in C^1(\cS,\R)$.
%\begin{remark} \rm
%	Lemma \ref{lemma_material_derivative_cartesian} shows that this definition of the material derivative equals the ones from \cite{Jankuhn_Reusken_Olshanksii,Gigaetal,Miura_17,Pruess_Simonett_book}.
%\end{remark}
\subsection{Time derivative of first fundamental form} \label{sectE}
In this section we consider a time derivative of the first fundamental form, which will be used in the remainder. The local coordinate system introduced in Section~\ref{seccoordinate}  depends on the time variable $t$, cf. \eqref{mapR} and Section~\ref{seccoordanddiffop}. In particular for the covariant basis $\bg_\alpha= \partial_\alpha R$ we have $\bg_\alpha= \bg_\alpha(\bxi,t)$, $\bxi \in U$, $t \in I$. Hence the first fundamental form, cf.  \eqref{firstform}, depends not only on  $\bxi$ but also on the time variable,  $g_{\alpha \beta}=g_{\alpha \beta}(\bxi,t)$. The change (as function of time) of the  metric tensor is determined by the velocity field $\bv$, which determines the time dependence of the parametrization $R=R(\bxi,t)=\Phi_t(\Phi_U(\bxi))$ via the flow map $\Phi_t$.
Using Theorem \ref{thm_partial_derivatives}, the following relation for the time derivative of the metric tensor is derived (recall $\bv = v_i \bg^i= v^i \bg_i$):
\begin{align}
	\frac{\partial}{\partial t} g_{\alpha\beta} &= \partial_t \bg_\alpha \cdot \bg_\beta + \bg_\alpha \cdot \partial_t \bg_\beta = \partial_t \partial_\alpha R \cdot \bg_\beta + \bg_\alpha \cdot \partial_t \partial_\beta R = \partial_\alpha \bv \cdot \bg_\beta + \partial_\beta \bv \cdot \bg_\alpha \nonumber \\
	& = \big( ( v_{\gamma|\alpha} - b_{\alpha\gamma} v_3 ) \bg^\gamma + ( v_{3|\alpha} + b_\alpha^\gamma v_\gamma ) \bg^3 \big)  \cdot \bg_\beta \nonumber \\
	&\hphantom{=} + \left( (v_{\gamma|\beta} - b_{\beta\gamma} v_3 ) \bg^\gamma + ( v_{3|\beta} + b_\beta^\gamma v_\gamma ) \bg^3 \right) \cdot \bg_\alpha  \nonumber \\
	&=v_{\beta|\alpha}  + v_{\alpha|\beta}  - 2 v_3 b_{\alpha\beta}.\label{eqE}
\end{align}
For this time derivative of the metric tensor, scaled with a factor $\frac12$, we  introduce the notation
\begin{equation} \label{defE}
E_{\alpha \beta}:=\frac12 \frac{\partial}{\partial t} g_{\alpha\beta}.
 \end{equation}
For a given $(\bxi,t) \in \cS$ a corresponding linear operator $\bE=\bE(\bxi,t): \R^3 \to \R^3$ is given by $\bE := E_{\alpha\beta} (\bg^\alpha \otimes \bg^\beta)=E^{\alpha\beta} (\bg_\alpha \otimes \bg_\beta)$. This operator can also be expressed in terms of the  covariant derivatives introduced in the Definitions~\ref{Def1} and \ref{Def2} as shown in the following lemma. A proof of this lemma is given in the Appendix.
\begin{lemma}\label{lemma_rate_strain_nabla_representation}
The following relations hold:
\begin{equation} \label{reslemma_rate_strain_nabla_representation}
\bE = \dfrac{1}{2} (\nablaG \bv + \nablaG^T \bv) = \frac{1}{2}  (\hatnablaG \bv + \hatnablaG^T \bv).
\end{equation}
\end{lemma}

\section{Derivations of surface Navier-Stokes equations} \label{secderivation}
In this section we outline five different  derivations of surface Navier-Stokes equations known from the literature \cite{Hu_Zhang,Jankuhn_Reusken_Olshanksii,Gigaetal,Miura_17,NitschkeReutherVoigt2019}, which  use both different physical principles and representations in different coordinate systems. In the five subsections below we present, in a unified framework, the following derivations:

\begin{enumerate}[(1)]
	\item \label{ansatz_hu} In \cite{Hu_Zhang} as physical principles the conservation laws of \emph{surface} mass and momentum quantities are used. Surface Navier-Stokes equations in \emph{curvilinear coordinates} are derived.
	\item \label{ansatz_reusken} In \cite{Jankuhn_Reusken_Olshanksii} the same  conservation laws of \emph{surface} mass and momentum quantities as in \cite{Hu_Zhang} are used and surface Navier-Stokes equations in \emph{Cartesian coordinates} in $\R^3$ are derived.
	\item \label{ansatz_giga} In \cite{Gigaetal} the same surface mass conservation law as in \cite{Hu_Zhang,Jankuhn_Reusken_Olshanksii} is used. Instead of a surface momentum conservation principle a \emph{variational energy principle} is used. The equations are derived in \emph{Cartesian coordinates} in $\R^3$.
	\item \label{ansatz_miura} In \cite{Miura_17} as physical principles the conservation laws of \emph{volume} mass and momentum quantities are used in a thin tubular neighborhood of the (evolving) surface. Combined with a thin film limit procedure, surface Navier-Stokes equations are derived in \emph{Cartesian coordinates}.
\item \label{ansatz_voigt} In \cite{NitschkeReutherVoigt2019} the same  physical principles of \emph{volume} mass and momentum conservation in a thin tubular neighborhood as in \cite{Miura_17} are used. The resulting volume Navier-Stokes equations are represented in a thin film \emph{curvilinear} local coordinate system. A thin film limit procedure is applied to derive \emph{tangential} surface Navier-Stokes equations in curvilinear coordinates.
\end{enumerate}

In the approaches \ref{ansatz_hu}, \ref{ansatz_reusken}, \ref{ansatz_miura}, \ref{ansatz_voigt} one uses an ansatz for the viscous stress tensor, namely the standard Newtonian tensor in \ref{ansatz_miura} and \ref{ansatz_voigt} and the Boussinesq-Scriven tensor in \ref{ansatz_hu} and \ref{ansatz_reusken}. In \ref{ansatz_giga} an ansatz for the viscous surface dissipation energy is used.
Below we outline only the key ideas of the derivations and refer to the corresponding papers for more details. % In the remainder of this paper, we delete the ''$\wedge$''-notation of the Cartesian differential operators. Since all occurring operator-valued functions are symmetric, no inconsistency arises concerning the surface divergence operators.

\subsection{Surface mass and momentum conservation in curvilinear coordinates} \label{secderivation1}
In this section, a derivation of surface Navier-Stokes equations along the same lines as in \cite{Hu_Zhang} is presented. In that paper the resulting surface Navier-Stokes equations are formulated in tensor calculus \emph{without} using surface differential operators like $\nabla_\Gamma$ and $\divG$. To be able to compare the resulting equations with those obtained in the other approaches, we rewrite these using the differential operators introduced in Section~\ref{sectlocal} and results derived in Section~\ref{sectE}.

The derivation is based on conservation laws of mass and momentum. We assume the surface to be inextensible, i.e. $\frac{d}{dt}\int_{\gamma(t)}1~ds=0$ holds for an arbitrary material subdomain $\gamma(t)\subset\G(t)$. The Leibniz rule \eqref{eqn_Leibniz_rule} and the arbitrariness of $\gamma(t)$ yield
\begin{equation} \label{div1}
\divG \bv=0.
\end{equation}
Let $\rho$ denote the surface mass density. Conservation of mass, the Leibniz rule and $\divG \bv=0$ lead to
\begin{equation*}
0 = \dfrac{d}{dt} \int_{\gamma(t)} \rho ~ds = \int_{\gamma(t)} \overset{\bullet}{\rho} ~ds.
\end{equation*}
Arbitrariness of $\gamma(t)$ and a smoothness assumption on $\rho$ imply $\overset{\bullet}{\rho} = 0$. Hence, if $\rho$ is constant on $\Gamma(0)$, which we assume here, it follows that the surface mass density $\rho$ is constant on the evolving surface $\Gamma(t)$.

As ansatz for surface momentum conservation the equation
\begin{align}\label{eqn_dynamic_eqn}
    \dfrac{d}{dt} \int_{\gamma(t)} \rho \bv ~ds = \bF(\gamma(t))
\end{align}
is used,
with a  force $\bF$ decomposed into external area forces acting on $\gamma(t)$ and internal forces acting on the boundary $\partial\gamma(t)$.
\begin{remark} \label{intrem}
 \rm The (surface) integral of a vector valued function $\int_{\gamma(t)}  \bu \, ds$, cf. \eqref{eqn_dynamic_eqn}, is defined in the usual way. We chose a fixed (not necessarily orthogonal) basis of $\R^3$, say $\bw_1,\bw_2,\bw_3$. For $\bu(s)= u^i(s) \bw_i$,  we then define $\int_{\gamma(t)} \bu(s) \, ds=  \bw_i \int_{\gamma(t)} u^i(s) \, ds $. The results derived below are independent of the choice of $\bw_1,\bw_2,\bw_3$.
\end{remark}

We collect the external forces, consisting of normal and shear stresses, in the force term $\bff=f^\alpha \bg_\alpha + f_N \bn$. For the internal forces the Cauchy ansatz is made,  i.e., these forces are of the form $\bT \nu$ with a stress tensor $\bT$ and $\nu$  the in-plane unit normal on $\partial\gamma(t)$.
Using  $\bT \nu = T^{\alpha\beta} \nu_\alpha \bg_\beta$ the total net force on $\gamma(t)$ can be written as
\begin{equation} \label{h1}
    \bF(\gamma(t)) = \int_{\gamma(t)} \bff ~ds + \int_{\partial\gamma(t)} T^{\alpha\beta} \nu_\alpha \bg_\beta~ds.
\end{equation}
 As in \cite{Hu_Zhang}, we apply the Leibniz rule on the left-hand side of \eqref{eqn_dynamic_eqn} and Greens formula on the boundary integral of the right-hand side of \eqref{h1}. Using $\divG \bv=0$, this yields
\begin{equation*}
\int_{\gamma(t)} \rho \overset{\bullet}{\bv} ~ds = \int_{\gamma(t)} \bff + T^{\beta \alpha}_{|\beta} \bg_\alpha + T^{\alpha\beta} b_{\alpha\beta} \bn ~ds.
\end{equation*}
Due to the arbitrariness of $\gamma(t)$, we obtain the following  system of surface partial differential equations (cf. \cite[equation (31)]{Hu_Zhang}):
\begin{equation} \label{NS1}
	\rho (\overset{\bullet}{\bv} \cdot \bg^\alpha) = f^\alpha + T^{\beta \alpha}_{|\beta}, \quad	\rho (\overset{\bullet}{\bv} \cdot \bn) = f_N + T^{\alpha\beta} b_{\alpha\beta},
\end{equation}
which consists of two equations for tangential velocity change $\overset{\bullet}{\bv} \cdot \bg^\alpha$ and one equation for velocity change in normal direction $\overset{\bullet}{\bv} \cdot \bn$.
As ansatz for the stress tensor $\bT$ the Boussinesq-Scriven form (in curvilinear coordinates)
\begin{equation} \label{NS2}
    T^{\alpha\beta} = -pg^{\alpha\beta} + 2\mu_0 E^{\alpha\beta}
\end{equation}
is used,
which involves the surface pressure $p$, the viscosity coefficient $\mu_0$ and  the time derivative of the metric tensor $E^{\alpha\beta}$, cf. \eqref{defE}.
\emph{The equations \eqref{div1}, \eqref{NS1} and \eqref{NS2} form the surface Navier-Stokes system derived in} \cite{Hu_Zhang}.

To be able to compare this surface Navier-Stokes system, which is formulated in terms of curvilinear coordinates, to equations derived in the sections below, we rewrite these equations using surface differential operators, cf. Definition~\ref{Def1}. From Lemma~\ref{lemma_rate_strain_nabla_representation} we obtain for the rate of strain tensor $\bE= E_{\alpha \beta}(\bg^\alpha \otimes \bg^\beta)$  the representation
\[ \bE= \bE(\bv) = \dfrac{1}{2} (\nablaG \bv + \nablaG^T \bv),\]
and thus the operator representation
\begin{equation}
\label{eq_Boussinesq-Scriven_form}
\bT = -p \bP + 2\mu_0 \bE
\end{equation}
for the stress tensor. From Lemma~\ref{lemma_divergence_local_representation} we get that the equations \eqref{NS1} can be rewritten as
\[
  \rho \overset{\bullet}{\bv} = \bff + \divG(\bT).
\]
Using this and the identity $\divG (p \bP) = \nablaG p + p \kappa \bn$, we obtain the following  representation of the surface Navier-Stokes system \eqref{div1}, \eqref{NS1}-\eqref{NS2} in terms of the surface differential operators as in Definition~\ref{Def1}:
\begin{align}
    \begin{cases}\label{Hu_Zhang_eqn_NS_nabla}
        \rho \overset{\bullet}{\bv} = \bff -\nablaG p - p \kappa \bn + 2 \mu_0 \divG \bE(\bv),\\
        \divG \bv = 0.
    \end{cases}
\end{align}

\subsection{Surface mass and momentum conservation in Cartesian coordinates}
We recall the model derived in \cite{Jankuhn_Reusken_Olshanksii}. It is based on the same fundamental laws of surface continuum mechanics as in the previous section. The formulation of the equations, however, is in Cartesian coordinates in $\R^3$. Hence, the surface differential operators ($\hatnablaG$ and $\hatdivG$) used are as in Definition~\ref{Def2}. The material derivative $\overset{\bullet}{\bv}$ is defined  in Cartesian coordinates as formulated in Lemma~\ref{lemma_material_derivative_cartesian}. Using the Leibniz rule, the inextensibility condition
$
\frac{d}{dt} \int_{\gamma(t)} 1 ~ds=0
$
yields
\begin{equation} \label{div2}
\hatdivG\bv=0.
\end{equation}
From mass conservation $
\frac{d}{dt} \int_{\gamma(t)} \rho ~ds=0$
we obtain, with the same arguments as in the previous section, that $\rho$ remains constant on $\Gamma(t)$ if it is constant on $\Gamma(0)$.
%\begin{align*}
%\overset{\bullet}{\rho}(\by,t)+\rho(\by,t)\divG\bv(\by,t)=\overset{\bullet}{\rho}(\by,t)=0.
%\end{align*}
We now consider the conservation of surface momentum, expressed  by the equation
\begin{equation} \label{mconservation}
\frac{d}{dt} \int_{\gamma(t)} \rho \bv ~ds = \int_{\gamma(t)} \bff \, ds +\int_{\partial\gamma(t)} \bff_{\nu} \, ds,
\end{equation}
with a contact force term $\bff_{\nu}$ on $\partial\gamma(t)$ and an area
force term $\bff$.
The integrals are defined as in Remark~\ref{intrem}.
As in the previous section, for the contact force term a Cauchy ansatz and Boussinesq-Scriven ansatz are used:
\begin{equation} \label{Contactf}
 \bff_\nu=\bT \nu, \quad \bT=-p\bP + 2\mu_0\bE(\bv), \quad \bE(\bv)= \frac12 (\hatnablaG \bv + \hatnablaG^T \bv).
\end{equation}
Lemma \ref{lemma_rate_strain_nabla_representation} shows that the definition of the rate of strain tensor $\bE$ equals the one from \cite{Hu_Zhang}, cf. equation \eqref{eq_Boussinesq-Scriven_form}.
 From the Stokes theorem and the identity $\hatdivG(p\bP)=\hatnablaG p + p \kappa \bn$, we obtain the  momentum balance for $\gamma(t)$:
\begin{align*}
\frac{d}{dt} \int_{\gamma(t)} \rho \bv \, ds = \int_{\gamma(t)}\bff -\hatnablaG p - p \kappa \bn+ 2 \mu_0 \hatdivG \bE(\bv) \, ds.
\end{align*}
Using the Leibniz rule and combining the result with \eqref{div2}, we obtain the following surface Navier-Stokes system:
\begin{align}
	\begin{cases}\label{J_R_O:full_Navier_Stokes}
		\rho \overset{\bullet}{\bv} = \bff- \hatnablaG p  - p \kappa \bn+ 2 \mu_0 \hatdivG \bE(\bv), \\
		\hatdivG \bv =0.
	\end{cases}
\end{align}
Based on Theorem~\ref{theorem_comparison_local_cartesian} we conclude that this PDE system is exactly the same as in \eqref{Hu_Zhang_eqn_NS_nabla}. This is not surprising, since the derivations of the two systems start from exactly the same physical principles.

\subsection{Energetic variational principle in Cartesian coordinates}
In this section, we summarize the variational approach presented in \cite{Gigaetal} to derive a surface Navier-Stokes system. This derivation is performed in Cartesian coordinates in $\R^3$. It is assumed that $\Gamma(t)$ is a closed surface.
First, in exactly the same way as in the sections above, inextensibility and mass conservation lead to the equation (in Cartesian coordinates)
\begin{align}
	\label{K_L_G:div_u=0}
	\hatdivG \bv &= 0,
\end{align}
and that $\rho$ is constant on $\cS$.
Instead of a momentum conservation ansatz as in \cite{Hu_Zhang,Jankuhn_Reusken_Olshanksii} (cf. equation \eqref{mconservation}) and an energetic variational approach based on the so-called  Least Action and  Minimum Dissipation Principles is used.
We outline the key steps.

The so-called action integral (``kinetic energy'') is defined by
\begin{equation*}
	A(\bx):=\int_0^T\int_{\G(t)} \dfrac{1}{2} \rho |\bv(\bx,t)|^2\,d\bx\,dt.
\end{equation*}
Recall that $\bx=\bx(\bz,t) \in \Gamma(t)$, $\bz \in \Gamma(0)$, are the particle trajectories and $\bv(\bx,t)$ the corresponding velocity fields, cf. \eqref{def_flow_map}.
Note that $\G(t)$ and $\bv$ are uniquely determined by the trajectories $\bx(\bz,t)$. The variation of the action integral with respect to $\bx(\bz,t)$ can be formally written as
\begin{equation*}
	D_{\bx}A(\bx)(\bw) = \int_0^T \int_{\Gamma(t)} F_{\textrm{cons}} \cdot \bw\, ds\,dt=:\langle F_{\textrm{cons}},\bw\rangle
\end{equation*}
for a ''suitable'' class of admissible velocities $\bw$. This relation defines a conservative force $F_{\textrm{cons}}$, cf. \cite{Xu_Sheng_Liu}. In \cite[Theorem 1.5]{Gigaetal} it is shown that under reasonable assumptions
\begin{equation}\label{K_L_G:F_cons}
	F_{\textrm{cons}}= -\rho \overset{\bullet}{\bv}
\end{equation}
holds. Another force, the so-called dissipation force, is derived from  variation of ''surface viscosity'' energy,  which is modeled by the  functional
\begin{equation}
   E_{\textrm{diss}}(\bv)=- \int_0^T \int_{\Gamma(t)} \mu_0 |\bE(\bv)|^2
\, ds\, dt,
\end{equation}
with viscosity coefficient $\mu_0$ and a strain tensor $\bE$ as in \eqref{Contactf}, cf. \cite{Gigaetal}. Variation with respect to the velocity field $\bv$ leads to the dissipation force
\begin{equation*}
   D_{\bv} E_{\textrm{diss}}(\bv)(\bw) = \langle F_{\textrm{diss}},\bw\rangle,
\end{equation*}
for a ''suitable'' class of admissible velocities $\bw$, cf. \cite{Xu_Sheng_Liu}. In \cite[Theorem 1.6]{Gigaetal} the relation
\begin{equation}\label{K_L_G:F_diss}
	F_{\textrm{diss}}= 2\mu_0 \hatdivG \bE(\bv)
\end{equation}
 is derived.
The Onsager principle (cf. \cite{Xu_Sheng_Liu,Onsager_I,Onsager_II}) states that the dynamics of  a system is determined by a competition between internal energy (here the kinetic energy) and dissipation. In our setting the corresponding equation is formally given by
\begin{equation}
    D_{\bx} A=-D_{\bv} E_{\textrm{diss}},
\end{equation}
cf. \cite[p. 385]{Gigaetal}. This implies $\langle F_{\textrm{cons}}+F_{\textrm{diss}},\bw\rangle =0$ for all admissible velocity fields $\bw$.  Due to the fact that we consider incompressible surface flows we restrict to velocity fields $\bw$ with $\hatdivG \bw=0$. The following corollary is based on \cite[Lemma 2.7]{Gigaetal}.
\begin{corollary}\label{K_L_G:cor_int_f_dot_phi=0}
Let $\mathbf{g} \in C(\cS)^3$ be such that $\langle \mathbf{g}, \bw\rangle = 0$ for all $\bw\in C^\infty(\cS)$ with $\hatdivG\bw = 0$. Then there exists $p\in C^1(\cS)$ such that
\begin{align*}
\mathbf{g} = \hatnablaG p + p\kappa\bn.
\end{align*}
\end{corollary}
Applying this corollary, we obtain $F_{\textrm{cons}}+F_{\textrm{diss}}= \hatnablaG p + p\kappa\bn$ for a suitable (pressure) function $p$. Combining this with \eqref{K_L_G:div_u=0}, \eqref{K_L_G:F_cons} and \eqref{K_L_G:F_diss} one obtains the following surface Navier-Stokes equations:
\begin{align}\label{K_L_G:full_Navier_Stokes}
	\begin{cases}
		\rho \overset{\bullet}{\bv} -2 \mu_0 \hatdivG \bE(\bv)= -\hatnablaG p - p\kappa \bn,\\
		\hatdivG \bv = 0.
	\end{cases}
\end{align}
Note that this system is exactly the same as in \eqref{J_R_O:full_Navier_Stokes}, if in the latter we restrict to the case without (outer) area forces, i.e. $\bbf=0$.
\subsection{Thin film approach in Cartesian coordinates} \label{sectMiura}
A different approach for deriving surface Navier-Stokes equations, based on a thin film limit procedure, is introduced in \cite{Miura_17}. It is assumed that $\Gamma(t)$ is a smoothly evolving closed surface. Around this surface a thin film domain $\Omega_\varepsilon(t):=\{\bx\in\R^3\,|\,\dist(\bx,\G(t))<\varepsilon \} $ with a sufficiently small $\varepsilon>0$ is defined, in which the incompressible three-dimensional Navier-Stokes equations with appropriate boundary conditions on $\partial \Omega_\varepsilon(t)$ are given. These equations describe mass and momentum conservation in the volume domain $\Omega_\varepsilon(t)$. One then studies the limit of the thickness going to zero and the resulting surface equations. In \cite{Miura_17} these limit equations are derived using formal asymptotic expansions (in the parameter $\varepsilon$). We outline a few key steps in the derivation and refer to \cite{Miura_17} for further explanations.

The signed distance function to $\Gamma(t)$ is denoted by $d(\cdot,t)$.  For $\varepsilon$ sufficiently small the closest point projection of $\bx \in \Omega_\varepsilon(t)$ is given by $\pi(\bx,t) = \bx - d(\bx,t) \bn(\bx,t)$.  We define the space-time domain $Q_{\varepsilon,I}$ and its boundary $\partial Q_{\varepsilon,I}$ by
\begin{align} \label{spacetime}
Q_{\varepsilon,I}:=\underset{t \in I}{\bigcup}\Omega_\varepsilon(t)\times \{t\},\qquad \partial Q_{\varepsilon,I}:=\underset{t \in I}{\bigcup}\partial\Omega_\varepsilon(t)\times \{t\}.
\end{align}
The unit outward normal vector $\bn_\varepsilon(\bx,t)$ and outward normal velocity $V_\varepsilon(\bx,t)$ on $\partial \Omega_\varepsilon$ are given by
\begin{align*}
	\bn_\varepsilon(\bx,t) = \begin{cases}
		~~\bn(\pi,t), ~\text{ if } d(\bx,t) = \varepsilon,\\
	 	 -\bn(\pi,t), ~ \text{ if } d(\bx,t) = -\varepsilon,
   	\end{cases} 	V_\varepsilon(\bx,t)=  \begin{cases}
		~~ V_\G(\pi,t),~ \text{ if } d(\bx,t) = \varepsilon,\\
	 	  -V_\G(\pi,t), ~ \text{ if } d(\bx,t) = -\varepsilon,
	\end{cases}
\end{align*}
with $\pi = \pi(\bx,t)$ and  $V_\G$ the normal velocity  of the surface $\G(t)$.

We consider an incompressible Navier-Stokes system in $Q_{\varepsilon,I}$ with (perfect slip) Navier boundary conditions:
\begin{equation} \label{Miura:Navier_Stokes_in_domain}
	\begin{aligned}
		\partial_t \bv_\varepsilon + (\bv_\varepsilon \cdot \hatnabla) \bv_\varepsilon +\hatnabla p_\varepsilon &= \mu_0  \widehat{\div} \left( \hatnabla \bv_\varepsilon \right) &&\text{ in } Q_{\varepsilon,I},\\
		\widehat{\div}\, \bv_\varepsilon &= 0  &&\text{ in } Q_{\varepsilon,I},\\
		\bv_\varepsilon \cdot \bn_\varepsilon &= V_\varepsilon && \text{ on } \partial Q_{\varepsilon,I},\\
		{[} \bE_3( \bv_\varepsilon ) \bn_\varepsilon ]_{tan} &= 0 && \text{ on } \partial Q_{\varepsilon,I},
	\end{aligned}
\end{equation}
where $[\ba]_{tan}$ denotes the tangential component to $\partial \Omega_\varepsilon(t)$ of a vector $\ba\in\R^3$ and $\bE_3(\bv) := \frac{1}{2}(\hatnabla \bv + \hatnablaT \bv )$ the strain tensor. We use the notation $\bE_3(\cdot)$ to distinguish this \emph{three}-dimensional strain tensor from the surface strain tensor $\bE(\cdot)$ used in the previous sections. The differential operators $\widehat{\div}$, $\hatnabla$, cf. Section~\ref{sectCartesian}, are the usual ones in $\R^3$ and $\partial_t$ is the usual time derivative. Note that here (following the presentation in \cite{Miura_17}) the density is scaled to $\rho=1$, but this is not essential.

The system defines the velocity $\bv_\varepsilon$ and pressure $p_\varepsilon$ of a  fluid in $Q_{\varepsilon,I}$. To derive equations defining the velocity of the fluid on $\G(t)$ only, consistent with \eqref{Miura:Navier_Stokes_in_domain} and depending only on values of functions on $\G(t)$, formal asymptotic expansions are \emph{assumed}. More precisely, it is assumed that for the solution pair $(\bv_\varepsilon,p_\varepsilon)$  there exist vector fields $\bv$, $\bv^1$, $\bv^2$ and scalar functions $p$, $p^1$ such that
\begin{subequations}\label{Miura:extension_both}
\begin{align}
	\label{Miura:extension_u}
	\bv_\varepsilon(\bx,t) &= \bv( \pi,t) + d(\bx,t) \bv^1(\pi,t) + d(\bx,t)^2 \bv^2(\pi,t) + r(d^3),\\
	\label{Miura:extension_p}
	p_\varepsilon(\bx,t) &=  p(\pi,t) + d(\bx,t) p^1(\pi,t) + r(d^2).
\end{align}
\end{subequations}
Here, $r(d^k) = r( d(\bx,t)^k)$ denotes a higher order term, cf. \cite{Miura_17}. The analysis in \cite{Miura_17} is not rigorous in the sense  that it is not clear whether or under which assumptions such expansions exist.

A key ingredient to obtain surface Navier-Stokes equations from the Navier-Stokes system in the thin film domain $Q_{\varepsilon,I}$ is the following lemma (cf. \cite[Lemma 2.7]{Miura_17}).
\begin{lemma}\label{Miura:Lemma:Derivation_of_extension}
Let $\phi$ be a scalar and $\bu$ a vector-valued function on $\cS$. The derivatives of the composite functions $\phi(\pi(\bx,t),t)$ and $\bu(\pi(\bx,t),t)$ with respect to $\bx$ and $t$ are of the form
\begin{align*}
\hatnabla \phi(\pi,t) &=  ( \hatnablaG \phi )(\pi,t) + d(\bx,t) [\bB \hatnablaG \phi ](\pi,t) + r(d(\bx,t)^2),\\
\partial_t \phi(\pi,t) &= \frac{d}{dt} \phi(\pi(\bx,t),t) + d(\bx,t) (\hatnablaG V_\G \cdot \hatnablaG \phi ) (\pi,t) + r(d(\bx,t)^2),
\end{align*}
and
\begin{align*}
	\hatnabla \bu(\pi,t) &= (\hatnabla \bu \bP )(\pi,t) + d(\bx,t) [\hatnabla \bu \bB](\pi,t) + r(d(\bx,t)^2),\\
	\partial_t \bu(\pi,t) &=  \frac{d}{dt} \bu(\pi(\bx,t),t) + d(\bx,t) [\hatnabla \bu \hatnablaG V_\G](\pi,t)+ r(d(\bx,t)^2),
\end{align*}
for $(\bx,t)\in Q_{\varepsilon,I}$ and with $\pi = \pi(\bx,t)$.
\end{lemma}

Substituting the  expansions \eqref{Miura:extension_both} in the Navier-Stokes equations, collecting zero and first order (in $\varepsilon$) terms and using Lemma \ref{Miura:Lemma:Derivation_of_extension}, the following result is derived in \cite[Section 4]{Miura_17}.

\begin{theorem} \label{thmmiura}
Let $\bv_\varepsilon$ and $p_\varepsilon$ satisfy the Navier-Stokes equations \eqref{Miura:Navier_Stokes_in_domain} in the moving domain $\Omega_{\varepsilon}(t)$ with given normal velocity $V_\G$. Then, the zeroth order velocity field $\bv$ and the zeroth and first order terms $p$ and $p^1$ satisfy the following equations on $\Gamma(t)$:
\begin{align}
\begin{cases}\label{Miura:full_Navier_Stokes}
	\bv \cdot \bn = V_\G, \\
	\overset{\bullet}{\bv} = -\hatnablaG p - p^1 \bn + 2 \mu_0 \hatdivG \bE(\bv), \\
	\hatdivG \bv = 0.
\end{cases}
\end{align}
Here the surface strain tensor $\bE$ is as in \eqref{Contactf}.
\end{theorem}

We briefly discuss the result of Theorem~\ref{thmmiura}. From the expansion \eqref{Miura:extension_both} we obtain
$
\bv = \bv_{\varepsilon|\G}$ and $p = p_{\varepsilon|\G}$ and thus the velocity field $\bv$ and pressure $p$ in \eqref{Miura:full_Navier_Stokes} coincide with
the bulk velocity and pressure ($\bv_\varepsilon$ and $p_\varepsilon$) evaluated on the surface. We indicate  why  in \eqref{Miura:full_Navier_Stokes} the first order term $p^1$ arises. From differentiation of the expansion \eqref{Miura:extension_p} we get for a fixed $t \in I$:
\begin{align}
\begin{split}\label{Miura:nabla_q_=_nabla_p_p1}
	\hatnabla p_\varepsilon(\bx,t)  &= \hatnabla [p(\pi,t)] + \hatnabla
[d(\bx,t) p^1(\pi,t)] + r(d)\\
	&= \hatnablaG p^e + \hatnabla d(\bx,t)~p^1(\pi,t) + d(\bx,t) \hatnabla[p^1(\pi,t)]  + r(d)\\
	&= \hatnablaG p^e + \bn(\bx,t) p^1(\pi,t) + r(d),
\end{split}
\end{align}
with $\pi = \pi(\bx,t)$ and $d = d(\bx,t)$.
For $\varepsilon \to 0$ we obtain the relation $\hatnabla p_\varepsilon = \hatnablaG p + p^1 \bn$. Hence, we expect  $\hatnablaG p + p^1 \bn$, and not only $\hatnablaG p$, to occur in \eqref{Miura:full_Navier_Stokes}. Analogously to \eqref{Miura:nabla_q_=_nabla_p_p1} we get, for $\varepsilon \to 0$, expressions for  $\hatnabla \bv_\varepsilon$ and $\widehat{\div}\left( \hatnabla \bv_\varepsilon \right)$ that contain
$\bv^1$ and $\bv^2$. The functions $\bv^1$ and $\bv^2$, however, do \emph{not} occur in the surface Navier-Stokes system \eqref{Miura:full_Navier_Stokes}. This is based on the relations (cf. \cite[Remark 4.5]{Miura_17})
\begin{equation*}
\bv^1=-(\nablaG \bv)\bn, \quad \bv^2=-\frac{1}{2}\left( \bB \nablaG \bv + \nablaG \bv^1 \right) \bn = 0.
\end{equation*}
Therefore, the resulting  system of equations \eqref{Miura:full_Navier_Stokes} has (only) the unknowns $\bv$, $p$ and $p^1$.
Comparing \eqref{Miura:full_Navier_Stokes}  with \eqref{J_R_O:full_Navier_Stokes} we see that instead of $-p \kappa \bn$ in \eqref{J_R_O:full_Navier_Stokes} we now have $p^1 \bn$ and an additional equation $\bv\cdot \bn = V_\Gamma$, with a given $V_\Gamma$. This additional scalar equation makes \eqref{Miura:full_Navier_Stokes} a closed system for the unknowns $\bv, p, p^1$.  Further comparison of the systems in \eqref{Miura:full_Navier_Stokes}  and \eqref{J_R_O:full_Navier_Stokes} is discussed in Section~\ref{secdiscussion}.
% \begin{align*}
% \bv^1 = -\bP(\hatnablaT \bv)\bn\quad\text{and} \quad \bv^2 = 0,
% \end{align*}
% . Therefore, opposite to the pressure (where $p^1$ occurs), the unknowns  $\bv^1$ and $\bv^2$ corresponding to the higher order terms in the $\bv_\varepsilon$-expansion can be eliminated and the resulting  system of equations \eqref{Miura:full_Navier_Stokes} has as unknowns $\bv$, $p$ and $p^1$.

\subsection{Thin film approach in curvilinear coordinates}
Similar to the modeling approach outlined in the previous subsection, the authors of \cite{NitschkeReutherVoigt2019} derive \emph{tangential} surface Navier-Stokes equations based on a thin-film limit procedure. Instead of using Cartesian coordinates a \emph{three-dimensional curvilinear} thin film coordinate system is used.
In the subsections below we outline this approach and the resulting surface Navier-Stokes equations.

\subsubsection{Thin film curvilinear coordinate system}
We consider an evolving surface as in Section~\ref{secmatsurface}. The evolution of the surface geometry is determined by the normal velocity field only. Therefore a surface parametrization can also be obtained by using the normal velocity field \cite{Bothe05}. More precisely, we consider the initial value problems \eqref{def_flow_map} with the velocity field $\bv$ replaced by $v_N \bn=(\bv \cdot \bn)\bn$ and a corresponding flow map (cf. Section~\ref{secmatsurface}) denoted by $\Phi_t^n$. Instead of the parametrization  in \eqref{mapR} we use $R_n(\bxi,t):= \Phi_t^n(\Phi_U(\bxi))$.
For a fixed $t \in I$ we consider  a thin film neighborhood $\Omega_\varepsilon(t)$ as defined in the previous section. A natural parametrization of this domain is given by
\begin{equation}
\tilde R_n(\bxi,\zeta)=\tilde R_n(\bxi,\zeta,t) := R_n(\bxi,t) + \zeta \bn(\bxi,t),
\end{equation}
with $\bxi \in U$,  $\zeta \in (-\varepsilon,\varepsilon)$.
%Since a constant thickness $h$ of the thin film with $\Gamma$ centered in the middle is assumed, we have
%\begin{equation*}
%\partial_t \tilde R_n |_{\partial \Omega_h} = v_N \bn.
Based on this thin film parametrization we introduce, analogous to Section~\ref{seccoordanddiffop}, curvilinear coordinates and representations of differential operators in these coordinates. Note that in
Section~\ref{seccoordanddiffop} we used a \emph{two}-dimensional surface parametrization with first fundamental form denoted by $g_{\alpha \beta}$, whereas in this section we have a \emph{three}-dimensional parametrization of the tubular domain $\Omega_\varepsilon(t)$.
Similar as in the previous sections, we use Greek letters to sum over $1,2$ and Latin letters to sum over $1,2,3$. Partial derivatives are denoted by $\partial_i$, i.e., $\partial_i=\frac{\partial}{\partial \xi_i}$, $i=1,2$, $\partial_3=\frac{\partial}{\partial \zeta}$. We introduce the covariant basis $\bG_i= \partial_i \tilde{R}_n$, the corresponding contravariant basis $\bG^i$, the metric tensor $G_{ij} := \bG_i \cdot \bG_j$ and the Christoffel symbols $\bGamma_{ij}^{k} := \frac{1}{2} G^{kl} ( \partial_i G_{jl} + \partial_j G_{il} - \partial_l G_{ij} )$.
Derivatives in curvilinear coordinates $(\bxi,\zeta)$ can be defined completely analogous to Section~\ref{sectlocal}. For a scalar function $\phi$ we define the gradient $\nabla \phi:=\partial_i \phi \bG^i$, for a vector field $\bu$ we define the (covariant) derivative $\nabla \bu= \partial_i \bu \otimes \bG^i$, the divergence $\div \bu :=\partial_i \bu \cdot \bG^i$ and for an operator valued function $\bT$ the divergence $\div \bT= (\partial_i \bT)^T \bG^i$. Using the fact that these operators do not depend on the choice of the parametrization, cf. \cite{Ciarlet2013}, one obtains the following relations with differential operators in Euclidean three-dimensional space, for which we used the $\widehat{~}$ notation, cf.~Section~\ref{sectCartesian}:
\begin{align} \label{correspondence}
\nabla \phi = \hatnabla \phi, ~~ \nabla  \bu = \hatnabla \bu, ~~ \div  \bu = \widehat{\div} \bu, ~~ \div \bT = \widehat{\div} \left( \bT^T \right),
\end{align}
with  $\widehat{\div}\bT:=\widehat{\div}
\big( \bT^T\hat{\be}_i \big) \hat{\be}_i$ the usual row-wise divergence of a tensor.
Below we use the notation without $\widehat{~}$.
%In the remainder, to be consistent with the notation in the previous section, we use the $\widehat{~}$ notation for these differential operators.
Analogous to Theorem \ref{thm_partial_derivatives}, cf. also equations \eqref{eq_partial_covariant_derivative_component}-\eqref{eq_divergence_vector_component}, one can represent these operators in terms of local components, e.g.:
\begin{align*}
 (\nabla \bu)_{ij} = u_{i|j} \left( \bG^i \otimes \bG^j \right), \quad \div  \bu =  u^i_{~|i}, \quad \div  \bT =  T^{~j}_{i~~|j} \bG^i,
\end{align*}
with
\begin{align*}
u_{i|j} := \partial_j u_i - \bGamma_{ij}^k u_k, \quad u^j_{~|i} := \partial_i u^j + \bGamma_{ki}^j u^k, \quad  T^{~j}_{i~~|k} := G^{jl} T_{il|k}, \\
T^{ij}_{~|k} := \partial_k T^{ij} + T^{lj} \bGamma_{lk}^i + T^{il} \bGamma_{lk}^j, \quad T_{ij|k} := \partial_k T_{ij} - T_{lj} \bGamma_{ik}^l - T_{il} \bGamma_{jk}^l.
\end{align*}

For deriving a limit equation in Section~\ref{sect_limit} below it is convenient to relate the three-dimensional metric tensor $G_{ij}$ to a suitable surface metric on $\Gamma(t)$. For the latter we use the one induced by the parametrization $R_n$.
With a slight abuse of notation, we use the same symbols as in Section \ref{seccoordinate}, e.g. $g_{\alpha \beta}$ for the metric tensor induced by $R_n$. One can  derive  the following useful results for these metric tensors \cite{NitschkeReutherVoigt2019}:
\begin{equation} \label{Gg}
\begin{split}
G_{\alpha \beta} &= g_{\alpha \beta} - 2 \zeta b_{\alpha \beta} + \zeta^2 b_{\alpha \gamma} b^\gamma_\beta,  \quad G_{\zeta\zeta} = 1 , \quad G_{\zeta \alpha} = G_{\alpha \zeta} = 0 , \\
G^{\alpha \beta} &= g^{\alpha \beta} + \cO(\zeta), \quad G^{\zeta \zeta} = 1, \quad G^{\zeta \alpha} = G^{\alpha \zeta} = 0, \quad \bGamma_{\alpha \beta}^\gamma = \Gamma_{\alpha \beta}^\gamma + \cO(\zeta), \\
\bGamma_{\alpha \beta}^\zeta &= b_{\alpha \beta} + \cO(\zeta), \quad \bGamma_{\alpha \zeta}^\beta = \bGamma_{\zeta \alpha}^\beta = - b_\alpha^\beta + \cO(\zeta), \quad \bGamma_{i \zeta}^\zeta = \bGamma_{\zeta i}^\zeta = \bGamma_{\zeta \zeta}^j = 0.
\end{split}
\end{equation}

The material derivative is defined as in Section~\ref{secdiffop_subsectmaterial} but now with respect to the parametrization $\tilde R_n(\bxi,\zeta,t)$:
\begin{align} \label{timeder}
\overset{\bullet}{f}(\by,t) :=\partial_t \bar{f}(\bxi,\zeta,t)= \partial_t {f}(\tilde R_n(\bxi,\zeta,t),t), \quad \by =\tilde R_n(\bxi,\zeta,t)\in \Omega_\varepsilon(t).
\end{align}
For the  velocity field corresponding to the parametrization $\tilde R_n$ we use the notation
\begin{equation} \label{defwr}
  \bw_R(\by,t):= \frac{\partial}{\partial t} \tilde R_n(\bxi,\zeta,t), \quad \by =\tilde R_n(\bxi,\zeta,t)\in \Omega_\varepsilon(t).
\end{equation}
Using  $\frac{\partial}{\partial t} R_n(\bxi,t) = (v_N \bn)(R_n(\bxi,t),t)$ it follows that $\bw_R= v_N\bn + \mathcal{O}(\zeta)$ holds.
The material derivative  can be reformulated in Cartesian coordinates as
\begin{equation}
\label{material_derivative_Voigt_surface}
\overset{\bullet}{f}(\by,t)  = \partial_t f(\by,t) + \nabla f (\by,t) \cdot \bw_R(\by,t), \quad \by \in \Omega_\varepsilon(t).
\end{equation}
\subsubsection{Navier-Stokes equation in thin film}
In \cite{NitschkeReutherVoigt2019} the authors derive a surface Ericksen-Leslie model, starting from a simplified local three-dimensional Ericksen-Leslie model (cf. \cite[equations (B1)-(B3)]{NitschkeReutherVoigt2019}). We simplify these equations by taking $\lambda=0$ in equation (B1). The resulting Navier-Stokes equations are similar to the ones in Section~\ref{sectMiura}. Note, however, that in that section we used Cartesian coordinates whereas in this section curvilinear thin film coordinates are used.  To simplify the notation we write $\bv$ instead of $\bv_\varepsilon$. The space-time domain is as defined in \eqref{spacetime}. The  thin film Navier-Stokes system from \cite{NitschkeReutherVoigt2019} is given by
\begin{equation} \label{Voigt_Navier_Stokes_in_domain}
	\begin{aligned}
		\partial_t  \bar \bv + \nabla^{\bu} \bv &= - \nabla  p_\varepsilon + \mu_0 \Delta \bv &&\text{ in } Q_{\varepsilon,I},\\
		{\div}\, \bv &= 0  &&\text{ in } Q_{\varepsilon,I},\\
		\bw_R \cdot \bn_\varepsilon &= \pm V_\Gamma && \text{ on } \partial Q_{\varepsilon,I}, \\
		{[} \bE_3( \bv) \bn_\varepsilon ]_{tan} &= 0 && \text{ on } \partial Q_{\varepsilon,I},
	\end{aligned}
\end{equation}
with $V_\Gamma$ the given normal velocity of $\Gamma(t)$, the Laplace operator $\Delta \bv := \div \nabla  \bv + \nabla \div  \bv$,  $[\cdot]_{tan}$  the tangential component to $\partial \Omega_\varepsilon(t)$ as in  \eqref{Miura:Navier_Stokes_in_domain} and $\bE_3(\bv) := \frac{1}{2}(\nabla \bv + \nabla^T \bv )$ the strain tensor. In \cite{NitschkeReutherVoigt2019} this strain tensor is denoted by the Lie derivative of the metric tensor,  $\cL_{\bv} \bG = \nabla \bv+ \nabla^T  \bv$.
The time derivative $\partial_t \bar \bv$ is defined in curvilinear coordinates as in \eqref{timeder}. The direction $\bu$ used in the directional derivative $\nabla^{ \bu}\bv = \nabla \bv\, \bu$ is the relative fluid velocity defined by
$  \bu := \bv -\bw_R$, with $\bw_R$ as in \eqref{defwr}. Using \eqref{material_derivative_Voigt_surface} and \eqref{correspondence} we obtain
\[
  \partial_t \bar \bv + \nabla^{\bu} \bv= \partial_t \bv + \nabla \bv \, \bw_R+ \nabla \bv (\bv - \bw_R)= \partial_t \bv + \nabla \bv \, \bv,
\]
an thus this is the usual material derivative in Cartesian coordinates, in particular the same as in \eqref{Miura:Navier_Stokes_in_domain}. In \cite{NitschkeReutherVoigt2019} the authors assume the thin film to evolve with constant thickness such that the surface is located in the middle of the outer domain all the time. Hence, $\bw_R|_{\partial \Omega_\varepsilon} = v_n \bn_\varepsilon$ and the resulting relative velocity $\bu$ is tangential to $\partial \Omega_\varepsilon$. Therefore, the third equations of \eqref{Voigt_Navier_Stokes_in_domain} and \eqref{Miura:Navier_Stokes_in_domain} coincide. We conclude that the two volume Navier-Stokes  systems  \eqref{Voigt_Navier_Stokes_in_domain} and \eqref{Miura:Navier_Stokes_in_domain} are equal.

 \subsubsection{Tangential surface Navier-Stokes system}
 \label{sect_limit}
Using the curvilinear coordinate system, a tangential limit system ($\varepsilon \downarrow 0$) of \eqref{Voigt_Navier_Stokes_in_domain} is derived in \cite{NitschkeReutherVoigt2019}. We sketch the key ingredients of the derivation.

The covariant components of the strain tensor are given by $\frac{1}{2} \left( v_{j|i} + v_{i|j} \right)$. The homogeneous Navier boundary condition can be rewritten as
\begin{equation}
\label{Voigt:Navier_Stokes_in_domain_BC}
v_{\alpha|\zeta} + v_{\zeta|\alpha}=0 \text{ on } \partial \Omega_\varepsilon(t).
\end{equation}
Using this, Taylor expansions and the results in \eqref{Gg}, the following relations can be derived (cf. \cite[equations (B9)-(B11), (B13)]{NitschkeReutherVoigt2019}):
\begin{align}
\label{eq_TaylorEstimates}
\begin{split}
v_{\zeta|\zeta} |_\Gamma &= \cO(\varepsilon^2), \quad (\bE_3(\bv))_{\alpha\zeta}|_\Gamma = \cO(\varepsilon^2),\\
\partial_\zeta (\bE_3(\bv))_{\alpha\zeta}|_\Gamma &= \cO(\varepsilon^2), \quad (\bE_3(\bv))_{\alpha\zeta|\zeta}|_\Gamma = \cO(\varepsilon^2).
\end{split}
\end{align}
On $\Gamma$ we denote the tangential component of the velocity by $\bv_T$, i.e.
\begin{equation*}
\bv_T = (v_T)_\alpha \bg^\alpha = (v_\alpha \bg^\alpha)|_{\Gamma}= \bP \bv|_{\Gamma} \in T^1\Gamma.
\end{equation*}
The following identity holds (cf. \cite[equation (B18)]{NitschkeReutherVoigt2019}):
\begin{equation}
\label{eq_connection_V_v}
v_{\alpha|\beta}|_\Gamma = (v_T)_{\alpha|\beta} - v_N b_{\alpha \beta}.
\end{equation}
We aim to derive equations for $\bv_T$ and $p=p_\varepsilon|_{\Gamma}$ on the surface. We first consider the second equation of \eqref{Voigt_Navier_Stokes_in_domain}. Using \eqref{eq_TaylorEstimates} and \eqref{eq_connection_V_v}, the following relation can be derived (cf. \cite[equation (B22)]{NitschkeReutherVoigt2019}):
\begin{equation}
\label{Voigt:Navier_Stokes_surface_div}
0 = (\div \bv)|_\Gamma = \divG \bv_T - v_N \kappa + \cO(\varepsilon^2).
\end{equation}
We now treat the projection of the material derivative in the first equation of \eqref{Voigt_Navier_Stokes_in_domain}. Using $\bu|_\Gamma=\bP \bv |_\Gamma$, \eqref{eq_connection_V_v} and $(v_T)^\beta  (v_T)_{\alpha|\beta} \bg^\alpha = (\nablaG  \bv_T) \bv_T =: \nablaG^{\bv_T}\bv_T$ we obtain for the tangential part of the directional derivative $\nabla^{\bu} \bv$:
\begin{align}\label{eqn_nabla_UV}
[\nabla^{\bu} \bv]_\alpha |_\Gamma & = u^i v_{\alpha|i} |_\Gamma = v^\beta v_{\alpha|\beta}|_\Gamma \nonumber \\
 & = (v_T)^\beta \left( (v_T)_{\alpha|\beta} - v_N b_{\alpha \beta} \right) = [\nablaG^{\bv_T}\bv_T - v_N \bB \bv_T]_\alpha.
\end{align}
Using $\partial_t \tilde{R}_n |_\Gamma = v_N \bn$ and the splitting $\bv = v^\alpha \partial_\alpha \tilde{R}_n + v_\zeta \bn_\varepsilon$, the following relation for the tangential component of the time derivative can be derived  (cf. \cite[equation (B24)]{NitschkeReutherVoigt2019}):
\begin{equation}
\label{eqn_dtV}
[\partial_t \bar \bv]_\alpha |_\Gamma = g_{\alpha\beta} \partial_t (\bar{v}_T)^\beta - v_N (b_{\alpha\beta} (v_T)^\beta + \partial_\alpha v_N ).
\end{equation}
From \eqref{eqn_nabla_UV} and \eqref{eqn_dtV}, we obtain (cf. \cite[equation (B26)]{NitschkeReutherVoigt2019}):
\begin{align}
\label{Voigt:Navier_Stokes_surface_material_derivative}
\bP \left.\left( \partial_t \bar \bv + \nabla^{ \bu}  \bv \right) \right|_\Gamma = (\partial_t (\bar{v}_T)^\alpha) \bg_\alpha + \nabla_\Gamma^{\bv_T} \bv_T - v_N (2\bB \bv_T + \nablaG v_N).
\end{align}
For the pressure term $p=p_\varepsilon|_{\Gamma}$ in \eqref{Voigt_Navier_Stokes_in_domain}, we get:
\begin{equation} \label{formp}
\bP (\nabla p_\varepsilon)|_\Gamma = \nablaG  p.
\end{equation}
Finally we consider the projection of the Laplacian in the first equation in \eqref{Voigt_Navier_Stokes_in_domain}.
For a solenoidal vector field we have
\begin{equation*}
\bP \left( \Delta \bv \right) |_\Gamma = \left( \left( \Delta \bv \right)_\alpha \bg^\alpha \right) |_\Gamma = 2 \left( \left(  \div \bE_3(\bv) \right)_\alpha \bg^\alpha \right) |_\Gamma.
\end{equation*}
Using \eqref{eq_TaylorEstimates}, the following relation can be derived (cf. \cite[equation (B17)]{NitschkeReutherVoigt2019}):
\begin{align*}
\left( \div \bE_3(\bv)   \right)_\alpha |_\Gamma & = g^{\beta \gamma} \left( \left( \left( \bE_3(\bv) \right)_{\alpha \gamma} |_\Gamma \right)_{|_\beta} - b_{\alpha \beta} \left( \bE_3(\bv) \right)_{\zeta \gamma} |_\Gamma \right) + \cO(\varepsilon^2) \\
  &= g^{\beta \gamma}  \left( \left( \bE_3(\bv) \right)_{\alpha \gamma} |_\Gamma \right)_{|_\beta}  + \cO(\varepsilon^2).
\end{align*}
Using $\left( \bE(\bv_T) \right)_{\alpha \gamma} = \frac{1}{2} \left( (v_T)_{\alpha|\gamma} + (v_T)_{\gamma|\alpha} \right)$, we get
\begin{align*}
\left( \bE_3(\bv) \right)_{\alpha \gamma}|_\Gamma = \frac{1}{2} \left( v_{\alpha|\gamma} + v_{\gamma|\alpha} \right) |_\Gamma &= \frac{1}{2} \left( (v_T)_{\alpha|\gamma} + (v_T)_{\gamma|\alpha} \right) - v_N b_{\alpha \gamma} \nonumber \\
&= \left( \bE(\bv_T) - v_N \bB \right)_{\alpha \gamma}.
\end{align*}
Combining these results we obtain
\begin{align}
\bP \left( \Delta \bv \right) |_\Gamma &= 2 \left( \left(  \div \bE_3(\bv) \right)_\alpha \bg^\alpha \right) |_\Gamma \nonumber  \\
&= 2 g^{\beta \gamma} \left( \left( \bE_3(\bv) \right)_{\alpha \gamma} |_\Gamma \right)_{|_\beta} \bg^\alpha + \cO(\varepsilon^2) \nonumber \\
&= 2 g^{\beta \gamma} \left( \bE(\bv_T) - v_N \bB \right)_{\alpha \gamma|\beta} \bg^\alpha + \cO(\varepsilon^2) \nonumber \\
&= 2 \left( \bE(\bv_T) - v_N \bB \right)_{~|\beta}^{\beta \mu} g_{\alpha \mu} \bg^\alpha + \cO(\varepsilon^2) \nonumber \\
&= 2 \left( \bE (\bv_T) - v_N \bB \right)_{~|\beta}^{\beta \mu} \bg_\mu + \cO(\varepsilon^2) \nonumber \\
&= 2 \bP \divG( \bE (\bv_T) - v_N \bB ) + \cO(\varepsilon^2). \label{eq_Laplace_div_lie}
\end{align}
Here, we used Lemma \ref{lemma_divergence_local_representation} in the last equation.
\begin{remark} \rm
\label{rem_Voigt_div}
Note that \eqref{eq_Laplace_div_lie} seems to differ from the first equation of (B21) from  \cite{NitschkeReutherVoigt2019}. However,  different definitions of the surface divergence operators for operator-valued functions are involved. Let $\tildedivG$ be the surface divergence operator used in \cite{NitschkeReutherVoigt2019}. For an operator-valued function $\bT$ the  relation
\begin{equation*}
\tildedivG \bT = \bP \divG \bT.
\end{equation*}
holds. Using this and $
2 \bE(\bv_T) = \nablaG \bv_T+ \nablaG^T  \bv_T
$ it follows that the first identity in \cite[equation (B21)]{NitschkeReutherVoigt2019} and equation \eqref{eq_Laplace_div_lie} coincide.
\end{remark}

Combining the results \eqref{Voigt:Navier_Stokes_surface_material_derivative}, \eqref{formp}, \eqref{eq_Laplace_div_lie}, \eqref{Voigt:Navier_Stokes_surface_div} and considering the thin film limit $\varepsilon \to 0$, we obtain the tangential Navier-Stokes equations on the surface in local coordinates (cf. \cite[equation (B27)-(B28)]{NitschkeReutherVoigt2019}):
\begin{equation} \label{Voigt:Navier_Stokes_surface}
\begin{aligned}
 (\partial_t (\bar{v}_T)^\alpha) \bg_\alpha + \nabla_\Gamma^{\bv_T} \bv_T - v_N (2\bB \bv_T + \nablaG v_N) &= - \nablaG p + 2 \mu_0 \bP \divG( \bE( \bv_T ) - v_N \bB ),\\
\divG \bv_T &= v_N \kappa.
\end{aligned}
\end{equation}
Using $\partial_t \bar \bv_T = (\partial_t (\bar{v}_T)^\alpha) \partial_\alpha R_n + (v_T)^\alpha \partial_\alpha \partial_t R_n$, we obtain for the tangential part of the time derivative
\begin{equation*}
(\partial_t \bar \bv_T)_\alpha = \partial_t \bar \bv_T \cdot \partial_\alpha R_n = g_{\alpha \beta} \partial_t (\bar{v}_T)^\beta - v_N b_{\alpha \beta} (v_T)^\beta.
\end{equation*}
Hence, the tangential surface Navier-Stokes equations \eqref{Voigt:Navier_Stokes_surface} can be rewritten as  (cf. \cite[equation (B30)-(B31)]{NitschkeReutherVoigt2019}):
\begin{align}
\label{Voigt:Navier_Stokes_surface_2}
\begin{cases}
\bP \partial_t \bar \bv_T + \nabla_\Gamma^{\bv_T} \bv_T - v_N (\bB \bv_T + \nablaG v_N) = - \nablaG p + 2 \mu_0 \bP \divG( \bE( \bv_T ) - v_N \bB ),\\
\divG \bv_T = v_N \kappa.
\end{cases}
\end{align}

\section{Discussion of  surface Navier-Stokes equations}\label{secdiscussion}
In this section we compare the different equations and  discuss a directional splitting in tangential and normal components. For the surface differential operators we use the ones without $\widehat{~~}$, but this is irrelevant, cf. Theorem~\ref{theorem_comparison_local_cartesian}.
As already mentioned above, the approaches (1), (2) and (3), cf. Section~\ref{secderivation}, result in the same system of surface Navier-Stokes equations, except that in (3) no source term $\mathbf f$ is considered. We recall the resulting equations, cf.   \eqref{Hu_Zhang_eqn_NS_nabla},  \eqref{J_R_O:full_Navier_Stokes} and \eqref{K_L_G:full_Navier_Stokes}, where for convenience we put $\rho=1$:
\begin{align}
    \begin{cases}\label{ResNS1}
     \overset{\bullet}{\bv} = \bff -\nablaG p - p \kappa \bn + 2 \mu_0 \divG \bE(\bv),\\
        \divG \bv = 0.
    \end{cases}
\end{align}
 The system resulting from ansatz \ref{ansatz_miura} is  different, cf. \eqref{Miura:full_Navier_Stokes}:
 \begin{align}
\begin{cases}\label{ResNS2}
	\bv \cdot \bn = V_\G, \\
	\overset{\bullet}{\bv} = -\nablaG p - p^1 \bn + 2 \mu_0 \divG \bE(\bv), \\
	\divG \bv = 0.
\end{cases}
\end{align}
 In this system an additional unknown scalar function $p^1$ appears. In the approach  \ref{ansatz_voigt} (only) a \emph{tangential} surface Navier-Stokes system is derived, given in \eqref{Voigt:Navier_Stokes_surface_2}, which we repeat here:
 \begin{equation} \label{ResNS3}
\begin{cases}
\bP \partial_t \bar \bv_T + \nabla_\Gamma \bv_T \bv_T  = - \nablaG  p + 2\mu_0 \bP \divG( \bE(\bv_T) -  v_N \bB )+ v_N (\bB  \bv_T + \nablaG v_N),\\
\divG  \bv_T = v_N \kappa.
\end{cases}
\end{equation}
In the following, for \eqref{ResNS1} and \eqref{ResNS2} we consider a splitting of the equations for $\bv = \bv_T + v_N\bn$ and $p$ in \emph{coupled} equations for  $ \bv_T$ , $p$ (``tangential surface Navier-Stokes'') and for $ v_N$ (normal velocity). In \cite{Jankuhn_Reusken_Olshanksii}  the relations (cf. \cite[Lemma 2.1, equation (3.9)]{Jankuhn_Reusken_Olshanksii}
\begin{equation}
\label{relA}
\begin{gathered}
	\bP \overset{\bullet}{\bv} = \overset{\bullet}{\bv_T} + ( \overset{\bullet}{\bn} \cdot \bv_T) \bn + v_N \overset{\bullet}{\bn}, \quad \overset{\bullet}{\bv}\cdot\bn = \overset{\bullet}{v_N}-\bv_T \cdot \overset{\bullet}{\bn}, \\
	\bn \cdot \divG \bE(\bv) = \tr(\bB \nablaG \bv_T) - v_N \tr(\bB^2),
\end{gathered}
\end{equation}
are derived. Using these we obtain the following splitting  of the surface Navier-Stokes equations \eqref{ResNS1} into (coupled) equations
\begin{equation}
\label{J_R_O:splitting_Navier_Stokes1}
\begin{aligned}
  \overset{\bullet}{\bv_T} &=   \bff_T-\nablaG p + 2\mu_0 \bP \divG \bE(\bv) -  \big((\overset{\bullet}{\bn}\cdot\bv_T\big) \bn + v_N \overset{\bullet}{\bn}),\\
\divG \bv_T &= v_N \kappa,
\end{aligned}
\end{equation}
for the surface pressure $p$ and tangential velocity $\bv_T$ and
\begin{equation}
\label{J_R_O:splitting_Navier_Stokes2}
\begin{aligned}
  \overset{\bullet}{v_N} &= f_N+ 2\mu_0 \bn \cdot \divG \bE(\bv) - p \kappa +  \overset{\bullet}{\bn}\cdot\bv_T\\
&= f_N+ 2 \mu_0 \big(\tr(\bB \nablaG \bv_T) - v_N\tr(\bB^2)\big) - p \kappa +  \overset{\bullet}{\bn} \cdot \bv_T,
\end{aligned}
\end{equation}
for the normal velocity $v_N$. We used  the splitting $\bff = \bff_T + f_N \bn$. Note that $\overset{\bullet}{\bv_T}$ denotes the material derivative (along $\bv$) of $\bv_T$ and not $(\overset{\bullet} \bv)_T$; similarly for $\overset{\bullet}{v_N}$. We call the system \eqref{J_R_O:splitting_Navier_Stokes1} \emph{tangential surface Navier-Stokes equations}. Note that in these equations the normal velocity $v_N$ occurs.
\begin{remark} \rm
The variational principle used in  \cite{Gigaetal} to derive system \eqref{ResNS1} (with $\bff=0$) also directly leads to a tangential surface Navier-Stokes system if the class of ''admissible'' velocities $\bw$ in the defining relations for  the force terms $F_{\textrm{cons}}$ and $F_{\textrm{diss}}$ is restricted to tangential ones, i.e. $\bP\bw=\bw$. This yields  tangential force terms $F_{\textrm{cons}}=-\rho \bP\overset{\bullet}{\bv}$ and $F_{\textrm{diss}}= 2\mu_0 \bP \divG \bE(\bv)$ and a tangential momentum equation that is the same as the first equation in \eqref{J_R_O:splitting_Navier_Stokes1} with $\bff_T = 0$.
% The following tangential Navier-Stokes equation can be derived:
% \begin{equation}
% \label{K_L_G:splitting_Navier_Stokes}
% 		\rho \bP \overset{\bullet}{\bv} + \hatnablaG p = 2 \mu_0 \bP \hatdivG \bE(\bv).
% \end{equation}
% Using \eqref{eq_compare_split_equations}, we observe that this tangential
% from ansatz \ref{ansatz_hu} and \ref{ansatz_reusken} if the force term vanishes ($\bff_T = 0$).
% Note that the procedure of deriving such tangential Navier-Stokes equations differ: while applying the projection $\bP$ to the full Navier-Stokes system \eqref{J_R_O:full_Navier_Stokes} to derive \eqref{J_R_O:splitting_Navier_Stokes}, equation \eqref{K_L_G:splitting_Navier_Stokes} is obtained
% by restricting to tangential ''admissible'' velocities $\bw$ in the defining relations of the force terms.
\end{remark}

%\begin{remark} \rm
 From the relation
\begin{equation}
 \label{eq_bullet_n}
 \overset{\bullet}{\bn} = - \bB \bv_T - \nablaG v_N
\end{equation}
 (cf. \cite[Lemma 2.2]{Jankuhn_Reusken_Olshanksii}) it follows  that
no $\frac{\partial}{\partial t}$ is involved in $\overset{\bullet}{\bn}$, which indicates
that the equation \eqref{J_R_O:splitting_Navier_Stokes2} determines the time dynamics of the normal velocity $v_N (\cdot,t)$, and thus of the surface $\G(t)$, whereas the tangential surface Navier-Stokes equations \eqref{J_R_O:splitting_Navier_Stokes1} determine the time dynamics of the tangential velocity $\bv_T(\cdot,t)$.
%\end{remark}

We now consider the  splitting of the Navier-Stokes system \eqref{ResNS2}. Applying the projection $\bP$ to the second equation in \eqref{ResNS2} the term $p^1 \bP \bn$ vanishes and the remaining terms are the same as in the projected version of the first equation in \eqref{ResNS1}. This implies that \eqref{ResNS2} results in \emph{the same tangential surface Navier-Stokes equations} as in \eqref{J_R_O:splitting_Navier_Stokes1} (with $\bff_T=0$). Taking the scalar product of the  second equation in \eqref{ResNS2} with $\bn$ and using the results \eqref{relA} one obtains
%
%that vUsing the same
%procedure as above the tangential component of $\bv$ and the pressure $p$ can be calculated by
%\begin{align*}
%\bP \overset{\bullet}{\bv} & = 2\mu_0 \bP\hatdivG \bE(\bv) - \hatnablaG p \\
%\Leftrightarrow \overset{\bullet}{\bv_T} & =  2\mu_0 \bP\hatdivG \bE(\bv) - \hatnablaG p - \overset{\bullet}{\bn}v_N - %(\overset{\bullet}{\bn}\cdot \bv_T)\bn .
%\end{align*}
%The normal component is given by
%\begin{align*}
%	\overset{\bullet}{\bv} \cdot \bn + \hatnablaG p \cdot \bn + p^1 (\bn \cdot \bn) &= 2 \mu_0 \hatdivG \bE(\bv)\cdot\bn \5\
%	\Leftrightarrow \overset{\bullet}{v_N} - \overset{\bullet}{\bn} \cdot \bv_T + p^1 &= 2 \mu_0 \hatdivG \bE(\bv)\cdot \bn \\
%	\Leftrightarrow
%\[ \overset{\bullet}{v_N} - \overset{\bullet}{\bn} \cdot \bv_T + p^1 &= 2 \mu_0 (\tr(\bB \hatnablaG \bv_T) - v_N %\tr(\bB^2) ).
%\end{align*}
%Thus we obtain the following splitting of \eqref{Miura:full_Navier_Stokes}:
%\begin{align}
%\begin{split}\label{Miura:Navier_Stokes_splitting}
%	\overset{\bullet}{\bv_T} &= 2\mu_0 \bP\hatdivG \bE(\bv) - \hatnablaG p - ((\overset{\bullet}{\bn}\cdot \bv_T)\bn+ v_N\overset{\bullet}{\bn}) ,\\
%	\hatdivG \bv_T &= v_N \kappa,\\
%
\begin{equation} \label{eqvN} \overset{\bullet}{v_N} = 2 \mu_0 ( \tr(\bB \nablaG \bv_T) - v_N \tr(\bB^2)) - p^1+ \overset{\bullet}{\bn} \cdot \bv_T ,
\end{equation}
i.e., similar to the normal velocity equation  \eqref{J_R_O:splitting_Navier_Stokes2}, but with $p \kappa$ replaced by the first-order unknown pressure function $p^1$ (and with $f_N=0$).
%Comparing the split Navier-Stokes systems \eqref{J_R_O:splitting_Navier_Stokes} and \eqref{Miura:Navier_Stokes_splitting}, we observe that the \emph{tangential equations are equal} if $\bff_T=0$ holds and the density is normalized to one $(\rho=1)$. The equations for the normal velocity, however, are different, even if $f_N=0$ and $\rho=1$.   The equation for the normal velocity
%in \eqref{Miura:Navier_Stokes_splitting} involves the first-order pressure function $p^1$.
From the first equation in \eqref{ResNS2}, with given $V_\Gamma$, one obtains the  normal velocity $v_N$, which can  be substituted in the tangential surface Navier-Stokes equations, which then determine $\bv_T$ and $p$. Given $v_N$ and $\bv_T$  the unknown $p^1$ is determined by \eqref{eqvN}.

Finally  we compare  the tangential Navier-Stokes equations \eqref{J_R_O:splitting_Navier_Stokes1}  with the tangential equations \eqref{ResNS3} that result  from ansatz \ref{ansatz_voigt}. Both systems contain the same equation $\divG \bv_T = v_N \kappa$, which results from the inextensibility condition. We now show that the two tangential momentum equations in  \eqref{J_R_O:splitting_Navier_Stokes1} and \eqref{ResNS3}   are also the same if $\bff_T=0$. This can be done as follows. First note that the material derivative $\overset{\bullet}{\bv_T}$ in \eqref{J_R_O:splitting_Navier_Stokes1} is in general \emph{not} tangential. Its normal component is balanced by the term $(\overset{\bullet}{\bn} \cdot \bv_T)\bn$ on the
right-hand side in \eqref{J_R_O:splitting_Navier_Stokes1}. This normal component can be eliminated by using the relations
$\bn \cdot \overset{\bullet}{\bv_T}= -\overset{\bullet}{\bn} \cdot \bv_T$, which follows from $\bn \cdot \bv_T=0$, and
\[
  \overset{\bullet}{\bv_T}= \bP \overset{\bullet}{\bv_T}+ (\bn \cdot \overset{\bullet}{\bv_T})\bn =\bP \overset{\bullet}{\bv_T}-(\overset{\bullet}{\bn} \cdot \bv_T)\bn.
\]
Using this, \eqref{eq_bullet_n} and $\bE(\bv) =  \bE(\bv_T) - v_N \bB$, the tangential momentum equation in  \eqref{J_R_O:splitting_Navier_Stokes1}, with $\bff_T=0$, can be rewritten as
\begin{equation} \label{Hk9}
\begin{split}
  \bP\overset{\bullet}{\bv_T} & = -\nablaG p + 2\mu_0 \bP \divG \bE(\bv) -   v_N \overset{\bullet}{\bn}\\
   & = -\nablaG p + 2\mu_0 \bP \divG \left(\bE(\bv_T) - v_N \bB \right)+ v_N \left( \bB \bv_T + \nablaG v_N \right).
  \end{split}
\end{equation}
The right-hand side of this equation is the same as the right-hand side in \eqref{ResNS3}. We now compare the material derivatives on the left-hand sides. Applying Lemma \ref{lemma_material_derivative_cartesian}, the left-hand side of \eqref{Hk9} yields
\begin{equation}\label{C1}
\bP \overset{\bullet}{\bv_T} = \bP (\partial_t \bv_T^e +  \nabla \bv_T^e \bv).
\end{equation}
For the left-hand side in  \eqref{ResNS3} we obtain, using \eqref{timeder} and \eqref{material_derivative_Voigt_surface},
\[
\begin{split}
\bP \partial_t \bar \bv_T + \nabla_\Gamma \bv_T \bv_T &= \bP (\partial_t \bar \bv_T + \nabla \bv_T^e \bv_T)\\
 &= \bP (\partial_t \bv_T^e + v_N  \nabla \bv_T^e \bn +  \nabla \bv_T^e \bv_T) \\
 &= \bP (\partial_t \bv_T^e +  \nabla \bv_T^e \bv),
\end{split}
\]
and comparing this with \eqref{C1} we observe that the material derivatives also coincide. Hence, we conclude that the two tangential momentum equations in  \eqref{J_R_O:splitting_Navier_Stokes1} and \eqref{ResNS3}   are  the same (for $\bff_T=0$).

In summary, we have shown that all five derivations (1)--(5) lead to \emph{the same tangential surface Navier-Stokes equations} \eqref{J_R_O:splitting_Navier_Stokes1}. The derivations (1)--(3) result in the same equation for the normal velocity, namely the one in \eqref{J_R_O:splitting_Navier_Stokes2}.

{\bf Acknowledgments.} The authors wish to thank the German Research Foundation (DFG) for financial support within the Research Unit ''Vector- and tensor valued surface PDEs'' (FOR 3013) with project no. RE 1461/11-1. The fruitful discussions with Elena Bachini and Veit Krause (both from TU Dresden) are acknowledged.

\bibliographystyle{siam}
\bibliography{literatur}{}

\section{Appendix}

We give a proof of the second equality in \eqref{extra}. The first equality can be derived in the same way.
\begin{proof}[of \eqref{extra}]
The product rule, \eqref{eqn_partial_deriv_u} and the symmetry of the Christoffel symbols yield
\begin{align*}
	\partial_\gamma\bT &= \partial_\gamma T_{\alpha\beta} (\bg^\alpha \otimes \bg^\beta) + T_{\alpha\beta} \left( (\partial_\gamma \bg^\alpha \otimes \bg^\beta) + (\bg^\alpha \otimes \partial_\gamma \bg^\beta) \right) \\
	&= \partial_\gamma T_{\alpha\beta} (\bg^\alpha \otimes \bg^\beta) + T_{\alpha\beta} \left( (-\Gamma_{\gamma\mu}^\alpha \bg^\mu + b_\gamma^\alpha \bg^3) \otimes \bg^\beta + \bg^\alpha \otimes (-\Gamma_{\gamma\mu}^\beta \bg^\mu + b_\gamma^\beta \bg^3) \right)\\
	&= \left( \partial_\gamma T_{\alpha\beta} - \Gamma_{\gamma\alpha}^\mu T_{\mu\beta} -  \Gamma_{\gamma\beta}^\mu T_{\alpha\mu} \right) (\bg^\alpha
\otimes \bg^\beta) + T_{\alpha\beta} b_\gamma^\alpha (\bg^3 \otimes \bg^\beta) + T_{\alpha\beta} b_\gamma^\beta (\bg^\alpha \otimes \bg^3)\\
	&= T_{\alpha\beta|\gamma} (\bg^\alpha \otimes \bg^\beta) + T_{\alpha\beta} b_\gamma^\alpha (\bg^3 \otimes \bg^\beta) + T_{\alpha\beta} b_\gamma^\beta (\bg^\alpha \otimes \bg^3).
	\end{align*}
\end{proof}

We give a proof of Lemma \ref{lemma_divergence_local_representation}.
\begin{proof}[of Lemma \ref{lemma_divergence_local_representation}]
We represent $\bT$ in local coordinates as $\bT = T^{\alpha\beta} (\bg_\alpha \otimes \bg_\beta)$. From the definition of the divergence and Theorem \ref{thm_partial_derivatives}, we get
\begin{align*}
	\divG \bT &= (\partial_\alpha \bT)^T \bg^\alpha = \left( T^{\gamma\beta}_{~|\alpha} (\bg_\gamma \otimes \bg_\beta) + T^{\gamma\beta} b_{\alpha\gamma} (\bg_3 \otimes \bg_\beta) + T^{\gamma\beta} b_{\alpha\beta} (\bg_\gamma \otimes \bg_3)\right)^T \bg^\alpha\\
	&= \left(T^{\gamma\beta}_{~|\alpha} (\bg_\beta \otimes \bg_\gamma) + T^{\gamma\beta} b_{\alpha\gamma} (\bg_\beta \otimes \bg_3) + T^{\gamma\beta} b_{\alpha\beta} (\bg_3 \otimes \bg_\gamma)\right) \bg^\alpha\\
	&= T^{\gamma\beta}_{~|\alpha} \bg_\beta (\underbrace{\bg_\gamma \cdot \bg^\alpha}_{\delta_\gamma^\alpha}) + T^{\gamma\beta} b_{\alpha\gamma} \bg_\beta (\underbrace{\bg_3 \cdot \bg^\alpha}_{=0}) + T^{\gamma\beta} b_{\alpha\beta} \bg_3 (\underbrace{\bg_\gamma \cdot \bg^\alpha}_{=\delta_\gamma^\alpha}) \\
	&= T^{\alpha\beta}_{~|\alpha} \bg_\beta + T^{\alpha\beta} b_{\alpha\beta} \bg_3.
\end{align*}
\end{proof}

\begin{lemma}\label{lemma_weingarten_cartesian}
The shape operator can be represented in Cartesian coordinates by $\bB = - \hatnablaG \bn^e$.
\end{lemma}
\begin{proof}
We use Theorem \ref{theorem_comparison_local_cartesian}, Theorem \ref{thm_partial_derivatives} and the symmetry of $b_{\alpha\beta}$ to derive
\begin{align*}
	- \hatnablaG \bn^e = - \nablaG \bn = - \nabla_\alpha \bn \otimes \bg^\alpha = - (\bP \partial_\alpha \bn) \otimes \bg^\alpha = (\bP b_{\alpha\beta} \bg^\beta) \otimes \bg^\alpha  = \bB.
\end{align*}
% 	Note the relations
%  \begin{align*}
% (\partial_\alpha\bn \cdot \bg_\beta) \bg^\beta = (\partial_\alpha\bn \cdot
% \bg_\beta) \bg^\beta + (\hspace*{-0.3cm}\underbrace{\partial_\alpha\bn \cdot \bn}_{=\frac{1}{2} \partial_\alpha(\bn \cdot \bn)=0 }\hspace*{-0.3cm})\bn = \partial_\alpha\bn.
% \end{align*}
% Using this and
% the representation of the Weingarten map  in covariant components \eqref{eqn_def_weingarten_covariant} yields {AR: avoid notation $\nabla_S$ in the proof; not needed}
% \begin{align*}
% \bB &= b_{\alpha\beta} (\bg^\alpha \otimes \bg^\beta) = (-\partial_\alpha\bn \cdot \bg_\beta) (\bg^\alpha \otimes \bg^\beta) = -\bg^\alpha \otimes \left[( \partial_\alpha \bn \cdot \bg_\beta) \bg^\beta \right] \\ & =
% -\bg^\alpha \otimes \partial_\alpha \bn
%  = -\bP \hatnablaT \bn= -(\hatnabla\bn\bP)^T.
% \end{align*}
% {AR: add info $-\bg^\alpha \otimes \partial_\alpha \bn
%  = -\bP \hatnablaT \bn$; here you will need $\bn^e$?}\\
% {AR: I think this proof does not have the right structure; the definition of $\widehat \nabla \bn^e$ (as mapping) is $\widehat \nabla \bn^e= \hat \partial_k \bn^e \otimes \hat \be_k$. Hence you should show that $b_{\alpha\beta} (\bg^\alpha \otimes \bg^\beta)=\widehat \partial_k \bn^e \otimes \hat \be_k$ holds...; cf. proof of theorem 3.6}\\
% In addition it holds $\hatnabla \bn \bP \bg_\alpha = \hatnabla \bn \bg_\alpha$ and $\hatnabla \bn \bP \bn = 0 = \hatnabla \bn \bn$. Hence,
% \begin{align*}
% \bB = -(\hatnabla \bn \bP)^T = - \hatnablaT \bn = - \hatnabla \bn
% \end{align*}
% follows with the symmetry of $\hatnabla \bn$.
\end{proof}

We give a proof of Lemma \ref{lemma_rate_strain_nabla_representation}
\begin{proof}[of Lemma \ref{lemma_rate_strain_nabla_representation}]
The second equality follows from the equality of the covariant gradients, cf. Theorem~\ref{theorem_comparison_local_cartesian}. We prove the first equality. Using \eqref{eq_covariant_derivative_component} and \eqref{eqn_partial_deriv_u}, we get
\begin{align*}
	\nablaG \bv_T + \nablaGT \bv_T &= (\partial_\alpha v_\beta - \G^\tau_{\alpha\beta} v_\tau) (\bg^\beta \otimes \bg^\alpha) + (\partial_\alpha v_\beta - \G^\tau_{\alpha\beta} v_\tau) (\bg^\alpha \otimes \bg^\beta).
\end{align*}
A direct calculation, using \eqref{eqn_partial_deriv_u}, yields
\begin{align*}
	\dfrac{1}{2} (\nablaG (v_N\bn) + \nablaGT (v_N\bn)) &= \dfrac{1}{2}\left( \bg^\alpha \otimes \bP\partial_\alpha (v_N\bn) + \bP \partial_\alpha (v_N\bn) \otimes \bg^\alpha \right) \\
	&= \dfrac{1}{2}\left( -v_N b_{\alpha\beta} (\bg^\alpha \otimes \bg^\beta) - v_N b_{\alpha\beta} (\bg^\beta \otimes \bg^\alpha) \right) \\
	&= - v_N b_{\alpha\beta} (\bg^\alpha \otimes \bg^\beta).
\end{align*}
Using these results and \eqref{eqE} we obtain
\begin{align*}
	\bE(\bv) &= E_{\alpha\beta}(\bg^\alpha \otimes \bg^\beta) = \left( \dfrac{1}{2} (v_{\alpha|\beta} + v_{\beta|\alpha}) - v_N b_{\alpha\beta} \right) (\bg^\alpha \otimes \bg^\beta)\\
	&= \left( \dfrac{1}{2} ( \partial_\beta v_\alpha - \G^\tau_{\alpha\beta}
v_\tau + \partial_\alpha v_\beta - \G^\tau_{\alpha\beta} v_\tau) -  v_N b_{\alpha\beta} \right) (\bg^\alpha \otimes \bg^\beta)\\
	&= \dfrac{1}{2} (\partial_\alpha v_\beta - \G^\tau_{\alpha\beta} v_\tau) (\bg^\beta \otimes \bg^\alpha) + \dfrac{1}{2} (\partial_\alpha v_\beta
- \G^\tau_{\alpha\beta} v_\tau) (\bg^\alpha \otimes \bg^\beta) - v_N b_{\alpha\beta} (\bg^\alpha \otimes \bg^\beta)\\
	&= \dfrac{1}{2}(\nablaG \bv_T + \nablaGT \bv_T) + \dfrac{1}{2} (\nablaG (v_N\bn) + \nablaGT (v_N\bn)) = \dfrac{1}{2}(\nablaG \bv + \nablaGT \bv),
	\end{align*}
which completes the proof.
\end{proof}

\end{document}